\documentclass[
a4paper,reqno]{amsart}
\usepackage[T1]{fontenc}
\usepackage[english]{babel}
\usepackage{amssymb,bbm,enumerate}
\usepackage{color}
\usepackage[linktocpage=true,colorlinks=true, linkcolor=blue, citecolor=red, urlcolor=green]{hyperref}

\renewcommand{\phi}{\varphi}

\newcommand{\mc}[1]{\mathcal{#1}}
\newcommand{\mf}[1]{\mathfrak{#1}}
\newcommand{\mb}[1]{\mathbb{#1}}
\newcommand{\id}{\mathbbm{1}}
\newcommand{\tint}{{\textstyle\int}}

\DeclareMathOperator{\Mat}{Mat}
\DeclareMathOperator{\Der}{Der}

\DeclareMathOperator{\sdeg}{sdeg}

\newcommand{\ass}[1]{\stackrel{#1}{\longleftrightarrow}}

\theoremstyle{plain}
\newtheorem{theorem}{Theorem}[section]
\newtheorem{lemma}[theorem]{Lemma}
\newtheorem{proposition}[theorem]{Proposition}
\newtheorem{corollary}[theorem]{Corollary}
\newtheorem{ansatz}[theorem]{Ansatz}

\theoremstyle{definition}
\newtheorem{definition}[theorem]{Definition}

\theoremstyle{remark}
\newtheorem{remark}[theorem]{Remark}

\numberwithin{equation}{section}

\definecolor{light}{gray}{.9}

\title{Dirac reduction for Poisson vertex algebras}

\author{Alberto De Sole, Victor G. Kac, Daniele Valeri}

\address{Dipartimento di Matematica, Sapienza Universit\`a di Roma,
P.le Aldo Moro 2, 00185 Rome, Italy.}
\email{desole@mat.uniroma1.it}

\address{Department of Mathematics, MIT,
77 Massachusetts Avenue, Cambridge, MA 02139, USA.}
\email{kac@math.mit.edu}

\address{SISSA, Via Bonomea 265, 34136 Trieste, Italy.}
\email{dvaleri@sissa.it}

\begin{document}

\pagestyle{plain}

\begin{abstract}
We construct an analogue of Dirac's reduction for an arbitrary 
local or non-local Poisson bracket
in the general setup of non-local Poisson vertex algebras.
This leads to Dirac's reduction of an arbitrary non-local Poisson structure.
We apply this construction to an example of a generalized Drinfeld-Sokolov hierarchy.
\end{abstract}

\maketitle

\tableofcontents

\section{Introduction}\label{sec:0}

Let $P$ be a Poisson algebra with Poisson bracket $\{\cdot\,,\,\cdot\}$,
let $\theta_1,\dots,\theta_m$ be some elements of $P$
such that the determinant of the matrix
$C=\big(\{\theta_i,\theta_j\}\big)_{i,j=1}^m$
is an invertible element of $P$,
and let $C^{-1}$ be the inverse matrix.
In his famous paper \cite{Dir50}
Dirac constructed a new bracket:
\begin{equation}\label{eq:01}
\{a,b\}^D=\{a,b\}
-\sum_{i,j=1}^m\{a,\theta_i\}(C^{-1})_{ij}\{\theta_j,b\}
\,.
\end{equation}
It is immediate to check that this new bracket is skewsymmetric
and satisfies the Leibniz rules,
and that the elements $\theta_i$ are central for this bracket,
i.e. $\{\theta_i,P\}^D=0$.
A remarkable observation of Dirac is that \eqref{eq:01}
satisfies the Jacobi identity,
hence it is a Poisson bracket on $P$.
This is important since the associative algebra ideal $I$ generated by the $\theta_i$'s
is a Poisson ideal for the bracket \eqref{eq:01},
hence this bracket defines a Poisson algebra structure on the factor algebra $P/I$,
called the \emph{Dirac reduction} of $P$ by the constraints $\theta_1,\dots,\theta_m$.

In the present paper we provide an analogous construction for (non-local)
Poisson vertex algebras.
This allows us to extend Dirac's construction to an arbitrary
local or non-local Poisson bracket.
(Note that even if the Poisson bracket we begin with is local,
the resulting Dirac's bracket is, in general, non-local.)

We apply our construction to obtain Dirac's reduction of a generalized
Drinfeld-Sokolov \cite{DS85} hierarchy studied in \cite{DSKV13},
thereby providing the resulting hierarchy with a bi-Hamiltonian structure.

Note that our formula \eqref{dirac}
coincides with Dirac's \eqref{eq:01} in the finite-dimensional case,
but in order to extend it to the infinite-dimensional case
the language of non-local $\lambda$-brackets
developed in \cite{DSK13} is indispensable.

A local Poisson vertex algebra is a local counterpart of the Beilinson-Drinfeld notion
of a Coisson Algebra \cite{BD04};
however it is unclear whether a non-local Poisson vertex algebra can be introduced in their
framework.
As far as we know there has been no rigorous discussion in the literature of general
non-local Poisson structures
and of the Dirac reduction in the infinite-dimensional setting,
see remarks in \cite{DSK13}.

Throughout the paper, unless otherwise specified, all vector spaces, tensor products etc.,
are defined over a field $\mb F$ of characteristic 0.

\subsubsection*{Acknowledgments}
We are grateful to Vsevolod Adler and Vladimir Sokolov
for stimulating correspondence.
We are also grateful to Takayuki Tsuchida
for an important observation which led us to the correction
of the first version of the paper. 
We wish to thank IHES, France,
where part of this research was conducted.
The third author was supported by the ERC grant
``FroM-PDE: Frobenius Manifolds and Hamiltonian Partial
Differential Equations''.

\section{Non-local Poisson vertex algebras}\label{sec:1}

We start by recalling the definition of a (non-local) Poisson vertex algebra,
following \cite{DSK13}, where one can find more details.

We use the following standard notation:
for a vector space $V$,
we let $V[\lambda]$, $V[[\lambda^{-1}]]$, and $V((\lambda^{-1}))$,
be, respectively,
the spaces of polynomials in $\lambda$,
of formal power series in $\lambda^{-1}$,
and of formal Laurent series in $\lambda^{-1}$, 
with coefficients in $V$.
Furthermore,
we shall use the following notation from \cite{DSK13}:
$$
V_{\lambda,\mu}:=V[[\lambda^{-1},\mu^{-1},(\lambda+\mu)^{-1}]][\lambda,\mu]\,,
$$
namely, the quotient of the $\mb F[\lambda,\mu,\nu]$-module
$V[[\lambda^{-1},\mu^{-1},\nu^{-1}]][\lambda,\mu,\nu]$
by the submodule 
$(\nu-\lambda-\mu)V[[\lambda^{-1},\mu^{-1},\nu^{-1}]][\lambda,\mu,\nu]$.
%
%
Recall that we have the natural embedding 
$\iota_{\mu,\lambda}:\,V_{\lambda,\mu}\hookrightarrow V((\lambda^{-1}))((\mu^{-1}))$
defined by expanding the negative powers of $\nu=\lambda+\mu$
by geometric series in the domain $|\mu|>|\lambda|$.

Let $\mc V$ be a differential algebra, i.e. a unital commutative associative algebra
with a derivation $\partial:\,\mc V\to\mc V$.
Recall that a \emph{(non-local)} $\lambda$-\emph{bracket} on $\mc V$ is a linear map
$\{\cdot\,_\lambda\,\cdot\}:\,\mc V\otimes \mc V\to \mc V((\lambda^{-1}))$
satisfying the following \emph{sesquilinearity} conditions:
\begin{equation}\label{20110921:eq1}
\{\partial a_\lambda b\}=-\lambda\{a_\lambda b\}
\,\,,\,\,\,\,
\{a_\lambda\partial b\}=(\lambda+\partial)\{a_\lambda b\}
\,,
\end{equation}
and the left and right \emph{Leibniz rules}:
\begin{equation}\label{20110921:eq3}
\begin{array}{l}
\{a_\lambda bc\}=b\{a_\lambda c\}+c\{a_\lambda b\}\,, \\
\{ab_\lambda c\}=\{a_{\lambda+\partial}c\}_\to b+\{b_{\lambda+\partial} c\}_\to a\,.
\end{array}
\end{equation}
Here and further an expression $\{a_{\lambda+\partial}b\}_\to c$ is interpreted as follows:
if $\{a_{\lambda}b\}=\sum_{n=-\infty}^Nc_n\lambda^n$, 
then $\{a_{\lambda+\partial}b\}_\to c=\sum_{n=-\infty}^Nc_n(\lambda+\partial)^nc$,
where we expand $(\lambda+\partial)^n$ in non-negative powers of $\partial$.
The (non-local) $\lambda$-bracket $\{\cdot\,_\lambda\,\cdot\}$ 
is called \emph{skewsymmetric} if
\begin{equation}\label{20110921:eq2}
\{b_\lambda a\}=-\{a_{-\lambda-\partial}b\}
\,\,\,\, \text{ for all } a,b\in\mc V
\,.
\end{equation}
The RHS of the skewsymmetry condition should be interpreted as follows:
we move $-\lambda-\partial$ to the left and
we expand its powers in non-negative powers of $\partial$,
acting on the coefficients on the $\lambda$-bracket.
The (non-local) $\lambda$-bracket $\{\cdot\,_\lambda\,\cdot\}$ 
is called \emph{admissible} if
\begin{equation}\label{20110921:eq4}
\{a_\lambda\{b_\mu c\}\}\in\mc V_{\lambda,\mu}
\qquad\forall a,b,c\in\mc V\,.
\end{equation}
Here we are identifying the space $\mc V_{\lambda,\mu}$
with its image in $\mc V((\lambda^{-1}))((\mu^{-1}))$ via the embedding $\iota_{\mu,\lambda}$.
Note that, if $\{\cdot\,_\lambda\,\cdot\}$ is a skewsymmetric admissible 
(non-local) $\lambda$-bracket on $\mc V$,
then we also have
$\{b_\mu\{a_\lambda c\}\}\in\mc V_{\lambda,\mu}$ 
and $\{\{a_\lambda b\}_{\lambda+\mu} c\}\in\mc V_{\lambda,\mu}$,
for all $a,b,c\in\mc V$ (see \cite[Rem.3.3]{DSK13}).
\begin{definition}\label{20130513:def}
A \emph{non-local Poisson vertex algebra} (PVA) is a differential algebra $\mc V$
endowed with a non-local $\lambda$-bracket,
$\{\cdot\,_\lambda\,\cdot\}:\,\mc V\otimes \mc V\to \mc V((\lambda^{-1}))$
satisfying
skewsymmetry \eqref{20110921:eq2},
admissibility \eqref{20110921:eq4},
and the following \emph{Jacobi identity}:
\begin{equation}\label{20110922:eq3}
\{a_\lambda\{b_\mu c\}\}-\{b_\mu\{a_\lambda c\}\}=\{\{a_\lambda b\}_{\lambda+\mu} c\}
\,\,\,\,\text{ for every } a,b,c\in\mc V\,,
\end{equation}
where the equality is understood in the space $\mc V_{\lambda,\mu}$.
In this case we call $\{\cdot\,_\lambda\,\cdot\}$ a (non-local) \emph{PVA} $\lambda$-\emph{bracket}.
\end{definition}

We shall often drop the term ``non-local'',
so when we will refer to PVA's and $\lambda$-brackets
we will always mean \emph{non-local PVA}'s
and \emph{non-local} $\lambda$-\emph{brackets}.
(This, of course, includes the local case as well.)

An element $\theta$ of a (non-local) PVA $\mc V$ is called \emph{central}
if $\{a_\lambda\theta\}=0$ for all $a\in\mc V$.
Note that, by skewsymmetry, this is equivalent to the condition that $\{\theta_\lambda a\}=0$
for all $a\in\mc V$.
Note also that, by the sesquilinearity and Leibniz rules,
a differential algebra ideal of $\mc V$ generated by central elements
is automatically a PVA ideal.

In \cite{DSK13} there have been proved the following two lemmas, which will be used in the following sections.
\begin{lemma}[{\cite[Lem.2.3]{DSK13}}]\label{20111006:lem}
Let $A(\lambda,\mu),B(\lambda,\mu)\in\mc V_{\lambda,\mu}$,
and let $S,T:\,\mc V\to \mc V$ be endomorphisms of $\mc V$ (viewed as a vector space). 
Then
$$
A(\lambda+S,\mu+T)B(\lambda,\mu)\in\mc V_{\lambda,\mu}\,,
$$
where we expand the negative powers of $\lambda+S$ and $\mu+T$ 
in non-negative powers of $S$ and $T$, acting on the coefficients of $B$.
In particular,
if $\mc V$ is a differential algebra, 
then the space $\mc V_{\lambda,\mu}$ is also a differential algebra, 
with the obvious product and action of $\partial$.
\end{lemma}
\begin{lemma}[{\cite[Lem.3.9]{DSK13}}]\label{20111012:lem}
Let $\{\cdot\,_\lambda\,\cdot\}:\,\mc V\times\mc V\to\mc V((\lambda^{-1}))$ 
be a $\lambda$-bracket on 
the differential algebra $\mc V$.
Suppose that 
$C(\partial)=\big(C_{ij}(\partial)\big)_{i,j=1}^\ell\in\Mat_{\ell\times\ell}\big(\mc V((\partial^{-1}))\big)$
is an invertible $\ell\times\ell$ matrix pseudodifferential operator with coefficients in $\mc V$,
and let
$C^{-1}(\partial)=\big((C^{-1})_{ij}(\partial)\big)_{i,j=1}^\ell
\in\Mat_{\ell\times\ell}\big(\mc V((\partial^{-1}))\big)$
be its inverse.
Letting $C_{ij}=\sum_{n=-\infty}^N c_{ij;n}\partial^n$,
the following identities hold for every $a\in\mc V$ and $i,j=1,\dots,\ell$:
\begin{equation}\label{20111012:eq2a}
\begin{array}{r}
\displaystyle{
\big\{a_\lambda (C^{-1})_{ij}(\mu)\big\}
=
-\sum_{r,t=1}^\ell\sum_{n=-\infty}^N
\iota_{\mu,\lambda}(C^{-1})_{ir}(\lambda+\mu+\partial)
} \\
\displaystyle{
\{a_\lambda c_{rt;n}\} (\mu+\partial)^n (C^{-1})_{tj}(\mu)
\,\in\mc V((\lambda^{-1}))((\mu^{-1}))
\,,
}
\end{array}
\end{equation}
and
\begin{equation}\label{20111012:eq2b}
\begin{array}{l}
\displaystyle{
\big\{(C^{-1})_{ij}(\lambda) _{\lambda+\mu} a\big\}
=
-\sum_{r,t=1}^\ell\sum_{n=-\infty}^N
\{{c_{rt;n}}_{\lambda+\mu+\partial}a\}_\to
} \\
\displaystyle{
\Big((\lambda+\partial)^n (C^{-1})_{tj}(\lambda)\Big)
\iota_{\lambda,\lambda+\mu}({C^*}^{-1})_{ri}(\mu) 
\,\in\mc V(((\lambda+\mu)^{-1}))((\lambda^{-1}))
\,,
}
\end{array}
\end{equation}
where $\iota_{\mu,\lambda}:\,\mc V_{\lambda,\mu}\to\mc V((\lambda^{-1}))((\mu^{-1}))$
and $\iota_{\lambda,\lambda+\mu}:\,
\mc V_{\lambda,\mu}\to\mc V(((\lambda+\mu)^{-1}))((\lambda^{-1}))$ 
are the natural embeddings defined above.
In equations \eqref{20111012:eq2a} and \eqref{20111012:eq2b},
$C(\lambda)\in\Mat_{\ell\times\ell}\mc V((\lambda^{-1}))$ 
denotes the symbol of the matrix pseudodifferential operator $C$
(i.e. the matrix with entries 
$C_{ij}(\lambda)=\sum_{n=-\infty}^N c_{ij;n}\lambda^n$),
and $C^*$ denotes its adjoint (its inverse being $(C^{-1})^*$).
\end{lemma}
\begin{corollary}\label{20130514:cor}
Let $\{\cdot\,_\lambda\,\cdot\}:\,\mc V\times\mc V\to\mc V((\lambda^{-1}))$ 
be a $\lambda$-bracket on the differential algebra $\mc V$.
Let
$C(\partial)=\big(C_{ij}(\partial)\big)_{i,j=1}^\ell\in\Mat_{\ell\times\ell}\big(\mc V((\partial^{-1}))\big)$
be an invertible $\ell\times\ell$ matrix pseudodifferential operator with coefficients in $\mc V$,
and let
$C^{-1}(\partial)=\big((C^{-1})_{ij}(\partial)\big)_{i,j=1}^\ell\in\Mat_{\ell\times\ell}\big(\mc V((\partial^{-1}))\big)$
be its inverse.
Let $a\in\mc V$, and assume that
\begin{equation}\label{20130514:eq1}
\{a_\lambda C_{ij}(\mu)\}\,\in\mc V_{\lambda,\mu}
\,\,\text{ for all } i,j=1,\dots,\ell\,.
\end{equation}
(As before, we identify $\mc V_{\lambda,\mu}$ with its image 
$\iota_{\mu,\lambda}(\mc V_{\lambda,\mu})\subset\mc V((\lambda^{-1}))((\mu^{-1}))$.)
Then, we have
$\big\{a_\lambda (C^{-1})_{ij}(\mu)\big\},
\,\big\{(C^{-1})_{ij}(\lambda) _{\lambda+\mu} a\big\}\,\in\mc V_{\lambda,\mu}$.
In fact, 
the following identities hold in the space $\mc V_{\lambda,\mu}$:
\begin{equation}\label{20111012:eq2c}
\big\{a_\lambda (C^{-1})_{ij}(\mu)\big\}
=
-\sum_{r,t=1}^\ell\sum_{n=-\infty}^N
(C^{-1})_{ir}(\lambda+\mu+\partial)
\{a_\lambda c_{rt;n}\} (\mu+\partial)^n (C^{-1})_{tj}(\mu)
\,,
\end{equation}
and
\begin{equation}\label{20111012:eq2d}
\big\{(C^{-1})_{ij}(\lambda) _{\lambda+\mu} a\big\}
=
-\!\!\sum_{r,t=1}^\ell
\!\sum_{n=-\infty}^N
\{{c_{rt;n}}_{\lambda+\mu+\partial}a\}_\to
({C^*}^{-1})_{ri}(\mu) 
(\lambda+\partial)^n (C^{-1})_{tj}(\lambda),
\end{equation}
where $C_{ij}=\sum_{n=-\infty}^N c_{ij;n}\partial^n$,
\end{corollary}
\begin{proof}
It is an immediate corollary of Lemmas \ref{20111006:lem} and \ref{20111012:lem}.
\end{proof}

\section{Dirac reduction for (non-local) PVAs}\label{sec:2}

Let $\mc V$ be a (non-local) Poisson vertex algebra 
with $\lambda$-bracket $\{\cdot_{\lambda}\cdot\}$.
Let $\theta_1,\ldots,\theta_m$ be elements of $\mc V$,
and let $\mc I=\langle\theta_1,\ldots,\theta_m\rangle_{\mc V}$
be the differential ideal generated by them.
If $\mc I\subset\mc V$ is a PVA ideal, then the quotient differential algebra $\mc V/\mc I$
inherits a natural structure of (non-local) PVA.
In general, we shall modify the PVA $\lambda$-bracket $\{\cdot_{\lambda}\cdot\}$
via a construction which, for the finite dimensional setup, was introduced by Dirac \cite{Dir50}.
We thus get a new PVA $\lambda$-bracket $\{\cdot_{\lambda}\cdot\}^D$ on $\mc V$,
with the property that all the elements $\theta_i$ are central with respect to the modified
$\lambda$-bracket.
Therefore, $\mc I\subset\mc V$ becomes a PVA ideal for the modified $\lambda$-bracket,
and so we can consider the quotient (non-local) PVA $\mc V/\mc I$.

Consider the matrix pseudodifferential operator
$$
C(\partial)=(C_{\alpha\beta}(\partial))_{\alpha,\beta=1}^m
\in\Mat_{m\times m}(\mc V((\partial^{-1})))
\,,
$$
whose symbol is
\begin{equation}\label{C}
C_{\alpha\beta}(\lambda)=\{\theta_{\beta}{}_{\lambda}\theta_{\alpha}\}\,.
\end{equation}
By the skew-commutativity axiom \eqref{20110921:eq2},
the pseudodifferential operator $C(\partial)$ is skewadjoint.
We shall assume that the matrix pseudodifferential operator $C(\partial)$ is invertible, 
and we denote its inverse by
$C^{-1}(\partial)=\big((C^{-1})_{\alpha\beta}(\partial)\big)_{\alpha,\beta=1}^m
\in\Mat_{m\times m}(\mc V((\partial^{-1})))$.
\begin{definition}\label{20130514:def}
The \emph{Dirac modification} of the PVA $\lambda$-bracket $\{\cdot_{\lambda}\cdot\}$,
associated to the elements $\theta_1,\dots,\theta_m$,
is the map
$\{\cdot_{\lambda}\cdot\}^D:\,\mc V\times\mc V\to\mc V((\lambda^{-1}))$
given by ($a,b\in\mc V$):
\begin{equation}\label{dirac}
\{a_{\lambda}b\}^D
=\{a_{\lambda}b\}
-\sum_{\alpha,\beta=1}^m
\{{\theta_{\beta}}_{\lambda+\partial}b\}_{\to}
(C^{-1})_{\beta\alpha}(\lambda+\partial)
\{a_{\lambda}\theta_{\alpha}\}\,.
\end{equation}
\end{definition}
\begin{theorem}\label{prop:dirac}
Let $\mc V$ be a (non-local) PVA with $\lambda$-bracket $\{\cdot\,_\lambda\,\cdot\}$.
Let $\theta_1,\dots,\theta_m\in\mc V$ be  elements such that
the corresponding matrix pseudodifferential operator 
$C(\partial)=(C_{\alpha\beta}(\partial))_{\alpha,\beta=1}^m\in\Mat_{m\times m}(\mc V((\partial^{-1})))$
given by \eqref{C} is invertible.
\begin{enumerate}[(a)]
\item
The Dirac modification $\{\cdot\,_\lambda\,\cdot\}^D$
given by equation \eqref{dirac}
is a PVA $\lambda$-bracket on $\mc V$.
\item
All the elements $\theta_i,\,i=1,\dots,m$, are central 
with respect to the Dirac modified $\lambda$-bracket:
$\{a_\lambda\theta_i\}^D=\{{\theta_i}_\lambda a\}^D=0$
for all $i=1,\dots,m$ and $a\in\mc V$.
\item
The differential ideal $\mc I=\langle\theta_1,\dots,\theta_m\rangle_{\mc V}\subset\mc V$,
generated by $\theta_1,\dots,\theta_m$,
is an ideal with respect to the Dirac modified $\lambda$-bracket $\{\cdot\,_\lambda\,\cdot\}^D$,
namely:
$\{\mc I\,_\lambda\,\mc V\}^D$,
$\{\mc V\,_\lambda\,\mc I\}^D\,
\subset\mc I((\lambda^{-1}))$.
\end{enumerate}
The quotient space $\mc V/\mc I$ is a (non-local) PVA,
with $\lambda$-bracket induced by $\{\cdot\,_\lambda\,\cdot\}^D$,
which we call the \emph{Dirac reduction} of $\mc V$ 
by the constraints $\theta_1,\dots,\theta_m$.
\end{theorem}
\begin{proof}
Both sesquilinearity conditions \eqref{20110921:eq1} 
for the Dirac modified $\lambda$-bracket \eqref{dirac} are immediate to check.
%
%
The skewsymmetry condition \eqref{20110921:eq2} 
for the Dirac modified $\lambda$-bracket \eqref{dirac}
can also be easily proved:
it follows by the skewsymmetry of the $\lambda$-bracket $\{\cdot\,_\lambda\,\cdot\}$,
and by the fact that the matrix $C(\partial)$ (hence $C^{-1}(\partial)$) is skewadjoint.
%
%
The Dirac modified $\lambda$-bracket $\{\cdot\,_\lambda\,\cdot\}^D$
obviously satisfies the left Leibniz rule \eqref{20110921:eq3},
since $\{\cdot\,_\lambda\,\cdot\}$ does,
and therefore it also satisfies the right Leibniz rule,
as a consequence of the left Leibniz rule and the skewsymmetry.


Next, we prove that the Dirac modified $\lambda$-bracket $\{\cdot\,_\lambda\,\cdot\}^D$
is admissible, in the sense of equation \eqref{20110921:eq4}.
For this, we compute the triple $\lambda$-bracket $\{a_\lambda\{b_\mu c\}^D\}^D$
using the definition \eqref{dirac},
the sesquilinearity conditions \eqref{20110921:eq1}
and the left and right Leibniz rules \eqref{20110921:eq3}.
We get
\begin{eqnarray}
&&
\big\{a_\lambda{\{b_\mu c\}^D}\big\}^D
=
\big\{
a_{\lambda}
\{b_{\mu}c\}
\big\}
\label{1a}\\
&&
-\sum_{\gamma,\delta=1}^m
\big\{
a_{\lambda}
\{{\theta_{\delta}}_{y}c\}
\big\}
\Big(\Big|_{y=\mu+\partial}
(C^{-1})_{\delta\gamma}(\mu+\partial)
\{b_{\mu}\theta_{\gamma}\}
\Big)
\label{2a}\\
&&
-\sum_{\gamma,\delta=1}^m
\{{\theta_{\delta}}_{\lambda+\mu+\partial}c\}_{\to}
\big\{
a_{\lambda}
(C^{-1})_{\delta\gamma}(y)
\big\}
\Big(\Big|_{y=\mu+\partial}
\{b_{\mu}\theta_{\gamma}\}
\Big)
\label{3a}\\
&&
-\sum_{\gamma,\delta=1}^m
\{{\theta_{\delta}}_{\lambda+\mu+\partial}c\}_{\to}
(C^{-1})_{\delta\gamma}(\lambda+\mu+\partial)
\big\{
a_{\lambda}
\{b_{\mu}\theta_{\gamma}\}
\big\}
\label{4a}\\
&&
-\sum_{\alpha,\beta=1}^m
\big\{
{\theta_{\beta}}_{\lambda+\partial}
\{b_{\mu}c\}
\big\}_{\to}
(C^{-1})_{\beta\alpha}(\lambda+\partial)
\{a_{\lambda}\theta_{\alpha}\}
\label{5a}\\
&&
+\sum_{\alpha,\beta,\gamma,\delta=1}^m
\big\{
{\theta_{\beta}}_{x}
\{{\theta_{\delta}}_{y}c\}
\big\}
\Big(\Big|_{x=\lambda+\partial}
(C^{-1})_{\beta\alpha}(\lambda+\partial)
\{a_{\lambda}\theta_{\alpha}\}
\Big)
\nonumber\\
&&
\,\,\,\,\,\,\,\,\,\,\,\,\,\,\,\,\,\,\,\,\,\,\,\,\,\,\,\,\,\,\,\,\,\,\,\,\,\,\,\,\,\,\,\,\,\,\,\,\,\,\,\,\,\,
\,\,\,\,\,\,\,\,\,
\Big(\Big|_{y=\mu+\partial}
(C^{-1})_{\delta\gamma}(\mu+\partial)
\{b_{\mu}\theta_{\gamma}\}
\Big)
\label{6a}\\
&&
+\sum_{\alpha,\beta,\gamma,\delta=1}^m
\{{\theta_{\delta}}_{\lambda+\mu+\partial}c\}_{\to}
\big\{
{\theta_{\beta}}_{x}
(C^{-1})_{\delta\gamma}(y)
\big\}
\nonumber\\
&&
\,\,\,\,\,\,\,\,\,\,\,\,\,\,\,\,\,\,
\Big(\Big|_{x=\lambda+\partial}
(C^{-1})_{\beta\alpha}(\lambda+\partial)
\{a_{\lambda}\theta_{\alpha}\}
\Big)
\Big(\Big|_{y=\mu+\partial}
\{b_{\mu}\theta_{\gamma}\}
\Big)
\label{7a}\\
&&
+\sum_{\alpha,\beta,\gamma,\delta=1}^m
\{{\theta_{\delta}}_{\lambda+\mu+\partial}c\}_{\to}
(C^{-1})_{\delta\gamma}(\lambda+\mu+\partial)
\big\{
{\theta_{\beta}}_{\lambda+\partial}
\{b_{\mu}\theta_{\gamma}\}
\big\}_{\to}
\nonumber\\
&&
\,\,\,\,\,\,\,\,\,\,\,\,\,\,\,\,\,\,\,\,\,\,\,\,\,\,\,\,\,\,\,\,\,\,\,\,\,\,\,\,\,\,\,\,\,\,\,\,\,\,\,\,\,\,
\,\,\,\,\,\,\,\,\,\,\,\,\,\,\,\,\,\,\,\,\,\,\,\,\,\,\,\,\,\,\,\,\,\,\,\,
(C^{-1})_{\beta\alpha}(\lambda+\partial)
\{a_{\lambda}\theta_{\alpha}\}
\,.\label{8a}
\end{eqnarray}
Here and further we use the following notation:
given an element 
$$
P(\lambda,\mu)=\sum_{m,n,p=-\infty}^N p_{m,n,p}\lambda^m\mu^n(\lambda+\mu)^p
\in\mc V_{\lambda,\mu}\,,
$$ 
and $f,g\in\mc V$, we let
$$
\begin{array}{l}
\displaystyle{
P(x,y)\Big(\Big|_{x=\lambda+\partial}f\Big)
\Big(\Big|_{y=\mu+\partial}g\Big)
} \\
\displaystyle{
=
\sum_{m,n,p=-\infty}^N
p_{m,n,p}(\lambda+\mu+\partial)^p
\big((\lambda+\partial)^mf\big)\big((\mu+\partial)^ng\big)
\,\,\in\mc V_{\lambda,\mu}
\,.
}
\end{array}
$$
All the terms \eqref{1a}, \eqref{2a}, \eqref{4a}, \eqref{5a}, \eqref{6a}, and \eqref{8a},
lie in $\mc V_{\lambda,\mu}$
by the admissibility assumption on $\{\cdot\,_\lambda\,\cdot\}$
and Lemma \ref{20111006:lem}.
Moreover, by the admissibility of $\{\cdot\,_\lambda\,\cdot\}$
and the definition \eqref{C} of the matrix $C(\partial)$,
condition \eqref{20130514:eq1} holds.
Hence, we can use Corollary \ref{20130514:cor} and Lemma \ref{20111006:lem}
to deduce that
the terms \eqref{3a} and \eqref{7a} lie in $\mc V_{\lambda,\mu}$ as well.
Therefore, 
$\{a_\lambda{\{b_\mu c\}^D}\}^D$ lies in $\mc V_{\lambda,\mu}$
for every $a,b,c\in\mc V$,
i.e. the Dirac modification $\{\cdot\,_\lambda\,\cdot\}^D$
is admissible.


In order to complete the proof of part (a) we are left to check the Jacobi identity \eqref{20110922:eq3}
for the Dirac modified $\lambda$-bracket.
We can use equation \eqref{20111012:eq2c} in Corollary \ref{20130514:cor}
to rewrite the terms \eqref{3a} and \eqref{7a}.
As a result, we get
\begin{eqnarray}
&&
\big\{a_\lambda{\{b_\mu c\}^D}\big\}^D
=
\big\{
a_{\lambda}
\{b_{\mu}c\}
\big\}
\label{1b}\\
&&
-\sum_{\gamma,\delta=1}^m
\big\{
a_{\lambda}
\{{\theta_{\delta}}_{y}c\}
\big\}
\Big(\Big|_{y=\mu+\partial}
(C^{-1})_{\delta\gamma}(\mu+\partial)
\{b_{\mu}\theta_{\gamma}\}
\Big)
\label{2b}\\
&&
+\sum_{\gamma,\delta,\eta,\zeta=1}^m
\{{\theta_{\delta}}_{\lambda+\mu+\partial}c\}_{\to}
(C^{-1})_{\delta \zeta}(\lambda+\mu+\partial)
\big\{
a_\lambda {\{{\theta_\eta}_y\theta_\zeta\}}
\big\}
\nonumber\\
&&
\,\,\,\,\,\,\,\,\,\,\,\,\,\,\,\,\,\,\,\,\,\,\,\,\,\,\,\,\,\,\,\,\,\,\,\,\,\,\,\,\,\,\,\,\,\,\,\,\,\,\,\,\,\,
\,\,\,\,\,\,\,\,\,
\Big(\Big|_{y=\mu+\partial}
(C^{-1})_{\eta\gamma}(\mu+\partial)
\{b_{\mu}\theta_{\gamma}\}
\Big)
\label{3b}\\
&&
-\sum_{\gamma,\delta=1}^m
\{{\theta_{\delta}}_{\lambda+\mu+\partial}c\}_{\to}
(C^{-1})_{\delta\gamma}(\lambda+\mu+\partial)
\big\{
a_{\lambda}
\{b_{\mu}\theta_{\gamma}\}
\big\}
\label{4b}\\
&&
-\sum_{\alpha,\beta=1}^m
\big\{
{\theta_{\beta}}_{\lambda+\partial}
\{b_{\mu}c\}
\big\}_{\to}
(C^{-1})_{\beta\alpha}(\lambda+\partial)
\{a_{\lambda}\theta_{\alpha}\}
\label{5b}\\
%
&&
+\sum_{\alpha,\beta,\gamma,\delta=1}^m
\big\{
{\theta_{\beta}}_{x}
\{{\theta_{\delta}}_{y}c\}
\big\}
\Big(\Big|_{x=\lambda+\partial}
(C^{-1})_{\beta\alpha}(\lambda+\partial)
\{a_{\lambda}\theta_{\alpha}\}
\Big)
\nonumber\\
&&
\,\,\,\,\,\,\,\,\,\,\,\,\,\,\,\,\,\,\,\,\,\,\,\,\,\,\,\,\,\,\,\,\,\,\,\,\,\,\,\,\,\,\,\,\,\,\,\,\,\,\,\,\,\,
\,\,\,\,\,\,\,\,\,
\Big(\Big|_{y=\mu+\partial}
(C^{-1})_{\delta\gamma}(\mu+\partial)
\{b_{\mu}\theta_{\gamma}\}
\Big)
\label{6b}\\
&&
-\sum_{\alpha,\beta,\gamma,\delta,\eta,\zeta=1}^m
\{{\theta_{\delta}}_{\lambda+\mu+\partial}c\}_{\to}
(C^{-1})_{\delta\zeta}(\lambda+\mu+\partial)
\big\{
{\theta_{\beta}}_x {\{{\theta_\eta}_y\theta_\zeta\}}
\big\}
\nonumber\\
&&
\,\,\,
\Big(\Big|_{x=\lambda+\partial}
(C^{-1})_{\beta\alpha}(\lambda+\partial)
\{a_{\lambda}\theta_{\alpha}\}
\Big)
\Big(\Big|_{y=\mu+\partial}
(C^{-1})_{\eta\gamma}(\mu+\partial)
\{b_{\mu}\theta_{\gamma}\}
\Big)
\label{7b}\\
&&
+\sum_{\alpha,\beta,\gamma,\delta=1}^m
\{{\theta_{\delta}}_{\lambda+\mu+\partial}c\}_{\to}
(C^{-1})_{\delta\gamma}(\lambda+\mu+\partial)
\big\{
{\theta_{\beta}}_{\lambda+\partial}
\{b_{\mu}\theta_{\gamma}\}
\big\}_{\to}
\nonumber\\
&&
\,\,\,\,\,\,\,\,\,\,\,\,\,\,\,\,\,\,\,\,\,\,\,\,\,\,\,\,\,\,\,\,\,\,\,\,\,\,\,\,\,\,\,\,\,\,\,\,\,\,\,\,\,\,
\,\,\,\,\,\,\,\,\,\,\,\,\,\,\,\,\,\,\,\,\,\,\,\,\,\,\,\,\,\,\,\,\,\,\,\,
(C^{-1})_{\beta\alpha}(\lambda+\partial)
\{a_{\lambda}\theta_{\alpha}\}
\,.\label{8b}
\end{eqnarray}
Exchanging the roles of $a$ and $b$ and of $\lambda$ and $\mu$
we get the second term of the LHS of the Jacobi identity:
\begin{eqnarray}
&&
\big\{b_\mu{\{a_\lambda c\}^D}\big\}^D
=
\big\{
b_{\mu}
\{a_{\lambda}c\}
\big\}
\label{1c}\\
&&
-\sum_{\gamma,\delta=1}^m
\big\{
b_{\mu}
\{{\theta_{\delta}}_{x}c\}
\big\}
\Big(\Big|_{x=\lambda+\partial}
(C^{-1})_{\delta\gamma}(\lambda+\partial)
\{a_{\lambda}\theta_{\gamma}\}
\Big)
\label{2c}\\
&&
+\sum_{\gamma,\delta,\eta,\zeta=1}^m
\{{\theta_{\delta}}_{\lambda+\mu+\partial}c\}_{\to}
(C^{-1})_{\delta \zeta}(\lambda+\mu+\partial)
\big\{
b_\mu {\{{\theta_\eta}_x\theta_\zeta\}}
\big\}
\nonumber\\
&&
\,\,\,\,\,\,\,\,\,\,\,\,\,\,\,\,\,\,\,\,\,\,\,\,\,\,\,\,\,\,\,\,\,\,\,\,\,\,\,\,\,\,\,\,\,\,\,\,\,\,\,\,\,\,
\,\,\,\,\,\,\,\,\,
\Big(\Big|_{x=\lambda+\partial}
(C^{-1})_{\eta\gamma}(\lambda+\partial)
\{a_{\lambda}\theta_{\gamma}\}
\Big)
\label{3c}\\
&&
-\sum_{\gamma,\delta=1}^m
\{{\theta_{\delta}}_{\lambda+\mu+\partial}c\}_{\to}
(C^{-1})_{\delta\gamma}(\lambda+\mu+\partial)
\big\{
b_{\mu}
\{a_{\lambda}\theta_{\gamma}\}
\big\}
\label{4c}\\
&&
-\sum_{\alpha,\beta=1}^m
\big\{
{\theta_{\beta}}_{\mu+\partial}
\{a_{\lambda}c\}
\big\}_{\to}
(C^{-1})_{\beta\alpha}(\mu+\partial)
\{b_{\mu}\theta_{\alpha}\}
\label{5c}
\end{eqnarray}
\begin{eqnarray}
&&
+\sum_{\alpha,\beta,\gamma,\delta=1}^m
\big\{
{\theta_{\beta}}_{y}
\{{\theta_{\delta}}_{x}c\}
\big\}
\Big(\Big|_{y=\mu+\partial}
(C^{-1})_{\beta\alpha}(\mu+\partial)
\{b_{\mu}\theta_{\alpha}\}
\Big)
\nonumber\\
&&
\,\,\,\,\,\,\,\,\,\,\,\,\,\,\,\,\,\,\,\,\,\,\,\,\,\,\,\,\,\,\,\,\,\,\,\,\,\,\,\,\,\,\,\,\,\,\,\,\,\,\,\,\,\,
\,\,\,\,\,\,\,\,\,
\Big(\Big|_{x=\lambda+\partial}
(C^{-1})_{\delta\gamma}(\lambda+\partial)
\{a_{\lambda}\theta_{\gamma}\}
\Big)
\label{6c}\\
&&
-\sum_{\alpha,\beta,\gamma,\delta,\eta,\zeta=1}^m
\{{\theta_{\delta}}_{\lambda+\mu+\partial}c\}_{\to}
(C^{-1})_{\delta\zeta}(\lambda+\mu+\partial)
\big\{
{\theta_{\beta}}_y {\{{\theta_\eta}_x\theta_\zeta\}}
\big\}
\nonumber\\
&&
\,\,\,
\Big(\Big|_{y=\mu+\partial}
(C^{-1})_{\beta\alpha}(\mu+\partial)
\{b_{\mu}\theta_{\alpha}\}
\Big)
\Big(\Big|_{x=\lambda+\partial}
(C^{-1})_{\eta\gamma}(\lambda+\partial)
\{a_{\lambda}\theta_{\gamma}\}
\Big)
\label{7c}\\
&&
+\sum_{\alpha,\beta,\gamma,\delta=1}^m
\{{\theta_{\delta}}_{\lambda+\mu+\partial}c\}_{\to}
(C^{-1})_{\delta\gamma}(\lambda+\mu+\partial)
\big\{
{\theta_{\beta}}_{\mu+\partial}
\{a_{\lambda}\theta_{\gamma}\}
\big\}_{\to}
\nonumber\\
&&
\,\,\,\,\,\,\,\,\,\,\,\,\,\,\,\,\,\,\,\,\,\,\,\,\,\,\,\,\,\,\,\,\,\,\,\,\,\,\,\,\,\,\,\,\,\,\,\,\,\,\,\,\,\,
\,\,\,\,\,\,\,\,\,\,\,\,\,\,\,\,\,\,\,\,\,\,\,\,\,\,\,\,\,\,\,\,\,\,\,\,
(C^{-1})_{\beta\alpha}(\mu+\partial)
\{b_{\mu}\theta_{\alpha}\}
\,.\label{8c}
\end{eqnarray}
In a similar way we compute the RHS of the Jacobi identity
for the Dirac modified $\lambda$-bracket,
using the definition \eqref{dirac},
the sesquilinearity \eqref{20110921:eq1},
the right Leibniz rule \eqref{20110921:eq3},
and equation \eqref{20111012:eq2d}
(recall that the matrix $C$ is skewadjoint).
We get
\begin{eqnarray}
&&
\big\{{\{a_\lambda b\}^D}_{\lambda+\mu}c\big\}^D
=
\big\{
{\{a_{\lambda}b\}}_{\lambda+\mu}c
\big\}
\label{1d}\\
&&
-\sum_{\alpha,\beta=1}^m
\big\{
\{{\theta_{\beta}}_{x}b\}
_{\lambda+\mu+\partial}
c\big\}_\to
\Big(\Big|_{x=\lambda+\partial}
(C^{-1})_{\beta\alpha}(\lambda+\partial)
\{a_{\lambda}\theta_{\alpha}\}
\Big)
\label{2d}\\
&&
-\sum_{\alpha,\beta,\eta,\zeta=1}^m
\big\{
{\{{\theta_\eta}_x\theta_\zeta\}}
_{\lambda+\mu+\partial}c
\big\}_\to
\Big(\Big|_{x=\lambda+\partial}
(C^{-1})_{\eta\alpha}(\lambda+\partial)
\{a_{\lambda}\theta_{\alpha}\}
\Big)
\nonumber\\
&&
\,\,\,\,\,\,\,\,\,\,\,\,\,\,\,\,\,\,\,\,\,\,\,\,\,\,\,\,\,\,\,\,\,\,\,\,\,\,\,\,\,\,\,\,\,\,\,\,\,\,\,\,\,\,
\,\,\,\,\,\,\,\,\,\,\,\,\,\,\,\,\,\,\,\,\,\,\,\,\,\,\,\,\,\,
\times
(C^{-1})_{\zeta\beta}(\mu+\partial) 
\{{\theta_{\beta}}_{-\mu-\partial}b\}
\label{3d}\\
&&
-\sum_{\alpha,\beta=1}^m
\big\{
{\{a_{\lambda}\theta_{\alpha}\}}_{\lambda+\mu+\partial}c
\big\}_\to
({C^*}^{-1})_{\alpha\beta}(\mu+\partial)
\{{\theta_{\beta}}_{-\mu-\partial}b\}
\label{4d}\\
&&
-\sum_{\gamma,\delta=1}^m
\{{\theta_{\delta}}_{\lambda+\mu+\partial}c\}_{\to}
(C^{-1})_{\delta\gamma}(\lambda+\mu+\partial)
\big\{
{\{a_{\lambda}b\}}_{\lambda+\mu}\theta_{\gamma}
\big\}
\label{5d}\\
&&
+\sum_{\alpha,\beta,\gamma,\delta=1}^m
\{{\theta_{\delta}}_{\lambda+\mu+\partial}c\}_{\to}
(C^{-1})_{\delta\gamma}(\lambda+\mu+\partial)
\big\{
\{{\theta_{\beta}}_{x}b\}_{\lambda+\mu+\partial}\theta_{\gamma}
\big\}_\to
\nonumber\\
&&
\,\,\,\,\,\,\,\,\,\,\,\,\,\,\,\,\,\,\,\,\,\,\,\,\,\,\,\,\,\,\,\,\,\,\,\,\,\,\,\,\,\,\,\,\,\,\,\,\,\,\,\,\,\,
\,\,\,\,\,\,\,\,\,\,\,\,\,\,\,\,\,\,\,\,\,\,\,\,\,\,\,
\Big(\Big|_{x=\lambda+\partial}
(C^{-1})_{\beta\alpha}(\lambda+\partial)
\{a_{\lambda}\theta_{\alpha}\}
\Big)
\label{6d}\\
&&
+\sum_{\alpha,\beta,\gamma,\delta,\eta,\zeta=1}^m
\{{\theta_{\delta}}_{\lambda+\mu+\partial}c\}_{\to}
(C^{-1})_{\delta\gamma}(\lambda+\mu+\partial)
\big\{
{\{{\theta_\eta}_x\theta_\zeta\}}_{\lambda+\mu+\partial}\theta_{\gamma}
\big\}_\to
\nonumber\\
&&
\,\,\,\,\,\,\,\,\,\,\,\,\,\,\,\,\,\,
\Big(\Big|_{x=\lambda+\partial}
(C^{-1})_{\eta\alpha}(\lambda+\partial)
\{a_{\lambda}\theta_{\alpha}\}
\Big)
(C^{-1})_{\zeta\beta}(\mu+\partial)
\{{\theta_{\beta}}_{-\mu-\partial}b\}
\label{7d}
\end{eqnarray}
\begin{eqnarray}
&&
+\sum_{\alpha,\beta,\gamma,\delta=1}^m
\{{\theta_{\delta}}_{\lambda+\mu+\partial}c\}_{\to}
(C^{-1})_{\delta\gamma}(\lambda+\mu+\partial)
\big\{
\{a_{\lambda}\theta_{\alpha}\}_{\lambda+\mu+\partial}\theta_{\gamma}
\big\}_\to
\nonumber\\
&&
\,\,\,\,\,\,\,\,\,\,\,\,\,\,\,\,\,\,\,\,\,\,\,\,\,\,\,\,\,\,\,\,\,\,\,\,\,\,\,\,\,\,\,\,\,\,\,\,\,\,\,\,\,\,
\,\,\,\,\,\,\,\,\,\,\,\,\,\,\,\,\,\,\,\,\,\,\,\,\,\,\,\,\,\,\,\,\,\,\,\,
({C^*}^{-1})_{\alpha\beta}(\mu+\partial)
\{{\theta_{\beta}}_{-\mu-\partial}b\}
\,.
\label{8d}
\end{eqnarray}
The following equations hold
due to the skewsymmetry \eqref{20110921:eq2},
the Jacobi identity \eqref{20110922:eq3},
and the fact that the matrix $C$ is skewadjoint:
$$
\begin{array}{l}
\text{RHS}\eqref{1b}-\text{RHS}\eqref{1c}=\text{RHS}\eqref{1d}
\,\,,\,\,\,\,
\eqref{2b}-\eqref{5c}=\eqref{4d}
\,,\\
\eqref{5b}-\eqref{2c}=\eqref{2d}
\,\,,\,\,\,\,
\eqref{4b}-\eqref{4c}=\eqref{5d}
\,\,,\,\,\,\,
\eqref{3b}-\eqref{8c}=\eqref{8d}
\,, \\
\eqref{8b}-\eqref{3c}=\eqref{6d}
\,\,,\,\,\,\,
\eqref{6b}-\eqref{6c}=\eqref{3d}
\,\,,\,\,\,\,
\eqref{7b}-\eqref{7c}=\eqref{7d}
\,.
\end{array}
$$
This concludes the proof of the Jacobi identity for the Dirac modified 
$\lambda$-bracket, and of part (a).


Note that the identities $C(\partial)C^{-1}(\partial)=C^{-1}(\partial)C(\partial)=1$ read, 
in terms of the symbols of the pseudodifferential operators $C(\partial)$ and $C^{-1}(\partial)$,
as
$$
\sum_{\beta=1}^m
\{{\theta_\beta}_{\lambda+\partial}\theta_\alpha\}_\to
(C^{-1})_{\beta\gamma}(\lambda)
=\delta_{\alpha,\gamma}
\,\,,\,\,\,\,
\sum_{\beta=1}^m
(C^{-1})_{\alpha\beta}(\lambda+\partial)
\{{\theta_\gamma}_\lambda\theta_\beta\}
=\delta_{\alpha,\gamma}
\,.
$$
Part (b) is an immediate consequence of these identities and the definition \eqref{dirac}
of the Dirac modified $\lambda$-bracket.
%
%
Part (c) is an obvious corollary of part (b), due to the sesquilinearity conditions
and the left and right Leibniz rules
for the Dirac modified $\lambda$-bracket.
%
%
Finally, the last statement of the theorem obviously follows.
\end{proof}
%

\subsection{Compatibility after Dirac reduction}\label{sec:2.3}

Recall that two PVA $\lambda$-brackets 
$\{\cdot\,_\lambda\,\cdot\}_0$ and $\{\cdot\,_\lambda\,\cdot\}_1$
on the same differential algebra $\mc V$
are said to be compatible if any of their linear combinations 
$\alpha\{\cdot\,_\lambda\,\cdot\}_0+\beta\{\cdot\,_\lambda\,\cdot\}_1$,
or, equivalently,  their sum $\{\cdot\,_\lambda\,\cdot\}_0+\{\cdot\,_\lambda\,\cdot\}_1$,
is again a PVA $\lambda$-bracket.
This amounts to the following admissibility condition ($a,b,c\in\mc V$):
\begin{equation}\label{20130516:eq1}
\big\{a_\lambda{\{b_\mu c\}_1}\big\}_0
+\big\{a_\lambda{\{b_\mu c\}_0}\big\}_1
\,
\in\mc V_{\lambda,\mu}
\,,
\end{equation}
and the Jacobi compatibility condition ($a,b,c\in\mc V$):
\begin{equation}\label{20130516:eq2}
\begin{array}{l}
\displaystyle{
\vphantom{\Big(}
\big\{a_\lambda{\{b_\mu c\}_1}\big\}_0-\big\{b_\mu{\{a_\lambda c\}_1}\big\}_0
-\big\{{\{a_\lambda b\}_0}_{\lambda+\mu} c\big\}_1
} \\
\displaystyle{
\vphantom{\Big(}
+\big\{a_\lambda{\{b_\mu c\}_0}\big\}_1-\big\{b_\mu{\{a_\lambda c\}_0}\big\}_1
-\big\{{\{a_\lambda b\}_1}_{\lambda+\mu} c\big\}_0=0
\,.}
\end{array}
\end{equation}

In general, if we have two compatible PVA $\lambda$-brackets
$\{\cdot\,_\lambda\,\cdot\}_0$ and $\{\cdot\,_\lambda\,\cdot\}_1$ on $\mc V$,
and we take their Dirac reductions by a finite number of constraints
$\theta_1,\dots,\theta_m$,
we do NOT get compatible PVA $\lambda$-brackets 
on $\mc V/\mc I$, where $\mc I=\langle\theta_1,\dots,\theta_m\rangle_{\mc V}$.
In Theorem \ref{20130516:thm1}
below we show
that the Dirac reduced PVA $\lambda$-brackets on $\mc V/\mc I$
are in fact compatible
in the special case when the constraints $\theta_1,\dots,\theta_m$
are central with respect to the first $\lambda$-bracket $\{\cdot\,_\lambda\,\cdot\}_0$.

\begin{theorem}\label{20130516:thm1}
Let $\mc V$ be a differential algebra,
endowed with two compatible PVA $\lambda$-brackets 
$\{\cdot\,_\lambda\,\cdot\}_0$, $\{\cdot\,_\lambda\,\cdot\}_1$.
Let $\theta_1,\dots,\theta_m\in\mc V$ be central elements
with respect to the first $\lambda$-bracket:
$\{a_\lambda\theta_i\}_0=0$ for all $i=1,\dots,m$, $a\in\mc V$.
Let $C(\partial)=\big(C_{\alpha,\beta}(\partial)\big)_{\alpha,\beta=1}^m$
be the matrix pseudodifferential operator
given by \eqref{C} for the second $\lambda$-bracket:
$C_{\alpha,\beta}(\lambda)=\{{\theta_\beta}_\lambda{\theta_\alpha}\}_1$.
Suppose that the matrix $C(\partial)$ is invertible,
and consider the Dirac modified PVA $\lambda$-bracket
$\{\cdot\,_\lambda\,\cdot\}_1^D$ given by \eqref{dirac}.
Then,
$\{\cdot\,_\lambda\,\cdot\}_0$
and $\{\cdot\,_\lambda\,\cdot\}_1^D$
are compatible PVA $\lambda$-brackets on $\mc V$.
Moreover, the differential algebra ideal
$\mc I=\langle\theta_1,\dots,\theta_m\rangle_{\mc V}$
is a PVA ideal for both the $\lambda$-brackets 
$\{\cdot\,_\lambda\,\cdot\}_0$
and $\{\cdot\,_\lambda\,\cdot\}_1^D$,
and we have the induced compatible PVA $\lambda$-brackets on $\mc V/\mc I$.
\end{theorem}
\begin{proof}
By assumption, $\{\cdot\,_\lambda\,\cdot\}_0$
is a PVA $\lambda$-bracket on $\mc V$,
and, by Theorem \ref{prop:dirac}(a),
$\{\cdot\,_\lambda\,\cdot\}_1^D$
is a PVA $\lambda$-bracket on $\mc V$ as well.
Hence, in order to prove the first assertion of the theorem,
we only need to check that they are compatible,
i.e. that they satisfy the admissibility condition \eqref{20130516:eq1}
and the Jacobi compatibility condition \eqref{20130516:eq2}.


Note that, since the elements $\theta_i$ are central 
with respect to $\{\cdot\,_\lambda\,\cdot\}_0$,
we have, for every $a\in\mc V$:
$$
\big\{a_\lambda{C_{\beta\alpha}(\mu)}\big\}_0
=\big\{a_\lambda{\{{\theta_\alpha}_\mu{\theta_\beta}\}_1}\big\}_0
=\big\{a_\lambda{\{{\theta_\alpha}_\mu{\theta_\beta}\}_1}\big\}_0
+\big\{a_\lambda{\{{\theta_\alpha}_\mu{\theta_\beta}\}_0}\big\}_1
\,,
$$
and the RHS lies in $\mc V_{\lambda,\mu}$
by the admissibility condition \eqref{20130516:eq1}.
Hence, condition \eqref{20130514:eq1} holds 
for the $\lambda$-bracket $\{\cdot\,_\lambda\,\cdot\}_0$,
and we can use Corollary \ref{20130514:cor}.
By the definition \eqref{dirac} of the Dirac modified 
$\lambda$-bracket $\{\cdot\,_\lambda\,\cdot\}_1^D$,
we therefore get, using sesquilinearity, the left and right Leibniz rules,
and equation \eqref{20111012:eq2c} 
for $\{\cdot\,_\lambda\,\cdot\}_0$:
\begin{eqnarray}
&&
\big\{a_\lambda{\{b_\mu c\}_1^D}\big\}_0
+\big\{a_\lambda{\{b_\mu c\}_0}\big\}_1^D=
\nonumber \\
&&
\big\{
a_\lambda{\{b_{\mu}c\}_1}
\big\}_0
+
\big\{
a_{\lambda}{\{b_\mu c\}_0}
\big\}_1
\label{a1}\\
&&
-\sum_{\alpha,\beta=1}^m
\big\{
a_\lambda
{\{{\theta_{\beta}}_{y}c\}_1}
\big\}_0
\Big(\Big|_{y=\mu+\partial}
(C^{-1})_{\beta\alpha}(\mu+\partial)
\{b_{\mu}\theta_{\alpha}\}_1
\Big)
\label{b1}\\
&&
+\sum_{\alpha,\beta,\gamma,\delta=1}^m
{\{{\theta_{\beta}}_{\lambda+\mu+\partial}c\}_1}_{\to}
(C^{-1})_{\beta\delta}(\lambda+\mu+\partial)
\big\{
a_\lambda {\{{\theta_\gamma}_y{\theta_\delta}\}_1}
\big\}_0 
\nonumber\\
&&
\,\,\,\,\,\,\,\,\,\,\,\,\,\,\,\,\,\,\,\,\,\,\,\,\,\,\,\,\,\,\,\,\,\,\,\,\,\,\,\,\,\,\,\,\,\,\,\,\,\,\,\,\,\,
\,\,\,\,\,\,\,\,\,\,\,\,
\times\Big(\Big|_{y=\mu+\partial}
(C^{-1})_{\gamma\alpha}(\mu+\partial)
\{b_{\mu}\theta_{\alpha}\}_1
\Big)
\label{c1}\\
&&
-\sum_{\alpha,\beta=1}^m
{\{{\theta_{\beta}}_{\lambda+\mu+\partial}c\}_1}_{\to}
(C^{-1})_{\beta\alpha}(\lambda+\mu+\partial)
\big\{
a_\lambda{\{b_{\mu}\theta_{\alpha}\}_1}
\big\}_0
\label{d1}\\
&&
-\sum_{\alpha,\beta=1}^m
{\big\{
{\theta_{\beta}}_{\lambda+\partial}{\{b_\mu c\}_0}
\big\}_1}_{\to}
(C^{-1})_{\beta\alpha}(\lambda+\partial)
\{a_{\lambda}\theta_{\alpha}\}_1
\,. \label{e1}
\end{eqnarray}
The term \eqref{a1} lies in $\mc V_{\lambda,\mu}$
by the admissibility assumption \eqref{20130516:eq1}.
Moreover, 
since the $\theta_i$'s are central for $\{\cdot\,_\lambda\,\cdot\}_0$,
we have, again by  \eqref{20130516:eq1}, that
$\big\{a_\lambda{\{{\theta_{\beta}}_{\mu}c\}_1}\big\}_0$
lies in $\mc V_{\lambda,\mu}$.
Therefore, by Lemma \ref{20111006:lem},
the term \eqref{b1} lies in $\mc V_{\lambda,\mu}$ as well.
With the same argument, we show that all the terms
\eqref{c1}, \eqref{d1} and \eqref{e1} lie in $\mc V_{\lambda,\mu}$.
Therefore, the admissibility condition \eqref{20130516:eq1}
holds for the pair of $\lambda$-brackets
$\{\cdot\,_\lambda\,\cdot\}_0$ and $\{\cdot\,_\lambda\,\cdot\}_1^D$.


Next, we prove the Jacobi compatibility condition \eqref{20130516:eq2}
for the pair of $\lambda$-brackets
$\{\cdot\,_\lambda\,\cdot\}_0$ and $\{\cdot\,_\lambda\,\cdot\}_1^D$.
Exchanging the roles of $a$ and $b$ and of $\lambda$ and $\mu$
in the above equation we get
\begin{eqnarray}
&&
\big\{b_\mu{\{a_\lambda c\}_1^D}\big\}_0
+\big\{b_\mu{\{a_\lambda c\}_0}\big\}_1^D=
\nonumber \\
&&
\big\{
b_\mu{\{a_{\lambda}c\}_1}
\big\}_0
+
\big\{
b_{\mu}{\{a_\lambda c\}_0}
\big\}_1
\label{a2}\\
&&
-\sum_{\alpha,\beta=1}^m
\big\{
b_\mu
{\{{\theta_{\beta}}_{x}c\}_1}
\big\}_0
\Big(\Big|_{x=\lambda+\partial}
(C^{-1})_{\beta\alpha}(\lambda+\partial)
\{a_{\lambda}\theta_{\alpha}\}_1
\Big)
\label{b2} \\
&&
+\sum_{\alpha,\beta,\gamma,\delta=1}^m
{\{{\theta_{\beta}}_{\lambda+\mu+\partial}c\}_1}_{\to}
(C^{-1})_{\beta\delta}(\lambda+\mu+\partial)
\big\{
b_\mu {\{{\theta_\gamma}_x{\theta_\delta}\}_1}
\big\}_0 
\nonumber\\
&&
\,\,\,\,\,\,\,\,\,\,\,\,\,\,\,\,\,\,\,\,\,\,\,\,\,\,\,\,\,\,\,\,\,\,\,\,\,\,\,\,\,\,\,\,\,\,\,\,\,\,\,\,\,\,
\,\,\,\,\,\,\,\,\,\,\,\,
\times\Big(\Big|_{x=\lambda+\partial}
(C^{-1})_{\gamma\alpha}(\lambda+\partial)
\{a_{\lambda}\theta_{\alpha}\}_1
\Big)
\label{c2}\\
&&
-\sum_{\alpha,\beta=1}^m
{\{{\theta_{\beta}}_{\lambda+\mu+\partial}c\}_1}_{\to}
(C^{-1})_{\beta\alpha}(\lambda+\mu+\partial)
\big\{
b_\mu{\{a_{\lambda}\theta_{\alpha}\}_1}
\big\}_0
\label{d2}\\
&&
-\sum_{\alpha,\beta=1}^m
{\big\{
{\theta_{\beta}}_{\mu+\partial}{\{a_\lambda c\}_0}
\big\}_1}_{\to}
(C^{-1})_{\beta\alpha}(\mu+\partial)
\{b_{\mu}\theta_{\alpha}\}_1
\,. \label{e2}
\end{eqnarray}
Furthermore,
a similar computation involving the definition \eqref{dirac} of the Dirac modified 
$\lambda$-bracket $\{\cdot\,_\lambda\,\cdot\}_1^D$,
the sesquilinearity conditions, the left and right Leibniz rules,
and equation \eqref{20111012:eq2d} 
for $\{\cdot\,_\lambda\,\cdot\}_0$,
gives
\begin{eqnarray}
&&
\big\{{\{a_\lambda b\}_0}_{\lambda+\mu} c\big\}_1^D
+\big\{{\{a_\lambda b\}_1^D}_{\lambda+\mu} c\big\}_0=
\nonumber \\
&&
\big\{
{\{a_\lambda b\}_0}_{\lambda+\mu}c
\big\}_1
+\big\{
{\{a_{\lambda}b\}_1}_{\lambda+\mu} c
\big\}_0
\label{a3}\\
&&
-\sum_{\alpha,\beta=1}^m
{\{{\theta_{\beta}}_{\lambda+\mu+\partial}c\}_1}_{\to}
(C^{-1})_{\beta\alpha}(\lambda+\mu+\partial)
\big\{
{\{a_\lambda b\}_0}_{\lambda+\mu}\theta_{\alpha}
\big\}_1
\label{b3}\\
&&
-\sum_{\alpha,\beta=1}^m
{\big\{
{\{{\theta_{\beta}}_{x}b\}_1}
_{\lambda+\mu+\partial} c
\big\}_0}_\to
\Big(\Big|_{x=\lambda+\partial}
(C^{-1})_{\beta\alpha}(\lambda+\partial)
\{a_{\lambda}\theta_{\alpha}\}_1
\Big)
\label{c3}\\
&&
+\sum_{\alpha,\beta,\gamma,\delta=1}^m
{\big\{
{\{{\theta_\gamma}_x{\theta_\delta}\}_1}_{\lambda+\mu+\partial}c
\big\}_0}_\to
\nonumber\\
&&
\times
\Big(\Big|_{x=\lambda+\partial}
(C^{-1})_{\gamma\alpha}(\lambda+\partial)
\{a_{\lambda}\theta_{\alpha}\}_1
\Big)
({C^*}^{-1})_{\delta\beta}(\mu+\partial)
{\{{\theta_{\beta}}_{-\mu-\partial}b\}_1}
\label{d3}\\
&&
-\sum_{\alpha,\beta=1}^m
{\big\{
{\{a_{\lambda}\theta_{\alpha}\}_1}
_{\lambda+\mu+\partial} c
\big\}_0}_\to
({C^*}^{-1})_{\alpha\beta}(\mu+\partial)
{\{{\theta_{\beta}}_{-\mu-\partial}b\}_1}
\,. \label{e3}
\end{eqnarray}
By the Jacobi compatibility condition \eqref{20130516:eq2}
we have
$$
\eqref{a1}-\eqref{a2}-\eqref{a3}=0\,.
$$
Moreover, by the skewadjointness of the matrix $C$ and by the skewsymmetry
of the $\lambda$-bracket $\{\cdot\,_\lambda\,\cdot\}_1$, we have
$$
\begin{array}{l}
\eqref{b1}-\eqref{e2}-\eqref{e3}=
\displaystyle{
-\sum_{\alpha,\beta=1}^m
\Big(
\big\{
a_\lambda
{\{{\theta_{\beta}}_{y}c\}_1}
\big\}_0
-
{\big\{
{\theta_{\beta}}_{y}{\{a_\lambda c\}_0}
\big\}_1}
} \\
\displaystyle{
-
{\big\{
{\{a_{\lambda}\theta_{\beta}\}_1}
_{\lambda+y} c
\big\}_0}
\Big)
\Big(\Big|_{y=\mu+\partial}
(C^{-1})_{\beta\alpha}(\mu+\partial)
\{b_{\mu}\theta_{\alpha}\}_1
\Big)
\,,}
\end{array}
$$
and this expression is zero by the Jacobi compatibility condition \eqref{20130516:eq2}
and the assumption that all the elements $\theta_i$'s are central 
with respect to $\{\cdot\,_\lambda\,\cdot\}_0$.
By similar arguments we conclude that
$$
\eqref{e1}-\eqref{b2}-\eqref{c3}=0
\,\,\text{ and }\,\,
\eqref{d1}-\eqref{d2}-\eqref{b3}=0
\,.
$$
Furthermore, again by equation \eqref{20130516:eq2}
and the fact that all the $\theta_i$'s are central 
with respect to $\{\cdot\,_\lambda\,\cdot\}_0$,
we get that, for all $\alpha,\beta=1,\dots,m$,
$\{{\theta_\alpha}_\lambda{\theta_\beta}\}_1$
is central with respect to $\{\cdot\,_\lambda\,\cdot\}_0$.
Therefore,
$$
\eqref{c1}=0
\,\,,\,\,\,\,
\eqref{c2}=0
\,\,,\,\,\,\,
\eqref{d3}=0
\,.
$$
In conclusion, 
the Jacobi compatibility condition \eqref{20130516:eq2}
holds for the pair of $\lambda$-brackets
$\{\cdot\,_\lambda\,\cdot\}_0$ and $\{\cdot\,_\lambda\,\cdot\}_1^D$.


Since, by assumption,
the elements $\theta_1,\dots,\theta_m$ are central with respect to $\{\cdot\,_\lambda\,\cdot\}_0$,
the differential ideal $\mc I$ generated by them is a PVA ideal
for this $\lambda$-bracket.
On the other hand, $\mc I$ is also a PVA ideal for $\{\cdot\,_\lambda\,\cdot\}_1^D$
by Theorem \ref{prop:dirac}(c).
The last assertion of the theorem follows.
\end{proof}

\section{Non-local Poisson structures and Hamiltonian equations}\label{sec:3}

\subsection{Algebras of differential functions}\label{sec:3.1}

Let $R_\ell=\mb F [u_i^{(n)}\, |\, i \in I,n \in \mb Z_+]$,
where $I=\{1,\dots,\ell\}$, be the algebra of differential polynomials
with derivation (uniquely) determined by $\partial(u_i^{(n)})=u_i^{(n+1)}$.
Recall from \cite{BDSK09} that an \emph{algebra of differential functions} 
in the variables $u_1,\dots,u_\ell$
is a differential algebra extension $\mc V$ of $R_\ell$
endowed with commuting derivations
$$
\frac{\partial}{\partial u_i^{(n)}}:\,\mc V\to\mc V
\,\,,\,\,\,\,
i \in I ,\,n \in \mb Z_+\,,
$$
extending the usual partial derivatives on $R_\ell$,
such that only a finite number of
$\frac{\partial f}{\partial u_i^{(n)}}$ are non-zero for each $f\in \mc V$,
and satisfying the following commutation relations:
\begin{equation}\label{eq:0.4}
\left[  \frac{\partial}{\partial u_i^{(n)}}, \partial \right] = \frac{\partial}{\partial u_i^{(n-1)}}
\,\,\,\,
\text{ (the RHS is $0$ if $n=0$) }\,.
\end{equation}
It is useful to write this commutation relation in terms of generating series:
\begin{equation}\label{eq:0.4b}
\sum_{n\in\mb Z_+} z^n  \frac{\partial}{\partial u_i^{(n)}}\circ \partial
=
(z+\partial)\circ \sum_{n\in\mb Z_+} z^n  \frac{\partial}{\partial u_i^{(n)}}
\,.
\end{equation}
We denote by $\mc C=\big\{c\in\mc V\,\big|\,\partial c=0\big\}\subset\mc V$ the subalgebra of \emph{constants},
and by
$$
\mc F=\Big\{f\in\mc V\,\Big|\,\frac{\partial f}{\partial u_i^{(n)}}=0
\,\,\text{ for all }\, i\in I,n\in\mb Z_+\Big\}\subset\mc V
$$
the subalgebra of \emph{quasiconstants}. It is easy to see that $\mc C\subset\mc F$.

Note that if $\mc V$ is an algebra of differential functions
and it is a domain,
then its field of fractions $\mc K$
is also an algebra of differential functions,
with the obvious extension of all the partial derivatives.

Recall that for $P\in\mc V^\ell$ we have the associated \emph{evolutionary vector field}
$$
X_P=\sum_{i\in I,n\in\mb Z_+}(\partial^nP_i)\frac{\partial}{\partial u_i^{(n)}}\,\in\Der(\mc V)\,.
$$
This makes $\mc V^\ell$ into a Lie algebra, with Lie bracket
$[X_P,X_Q]=X_{[P,Q]}$, given by
$$
[P,Q]=X_P(Q)-X_Q(P)
=D_Q(\partial)P-D_P(\partial)Q
\,,
$$
where $D_P(\partial)$ and $D_Q(\partial)$ denote the Frechet derivatives of $P,Q\in\mc V^\ell$.

In general,
for $\theta=\big(\theta_\alpha\big)_{\alpha=1}^m\in\mc V^m$,
the \emph{Frechet derivative} 
$D_{\theta}(\partial)\in\Mat_{m\times\ell}\mc V[\partial]$ 
is defined by
\begin{equation}\label{20120126:eq2}
D_{\theta}(\partial)_{\alpha i}=
\sum_{n\in\mb Z_+}\frac{\partial \theta_\alpha}{\partial u_i^{(n)}}\partial^n
\,\,,\,\,\,\,
\alpha=1,\dots,m\,,\,\,i=1,\dots,\ell
\,.
\end{equation}
Its adjoint $D_{\theta}^*(\partial)\in\Mat_{\ell\times m}\mc V[\partial]$ is then given by
$$
D_{\theta}^*(\partial)_{i\alpha}=
\sum_{n\in\mb Z_+}(-\partial)^n
\frac{\partial \theta_\alpha}{\partial u_i^{(n)}}
\,\,,\,\,\,\,
\alpha=1,\dots,m\,,\,\,i=1,\dots,\ell
\,.
$$

\subsection{Rational matrix pseudodifferential operators}\label{sec:3.2}

Let $\mc V$ be a differential algebra with derivation $\partial$.
We assume that $\mc V$ is a domain, and we denote by $\mc K$
its field of fractions.
Consider the skewfield $\mc K((\partial^{-1}))$ of pseudodifferential operators 
with coefficients in $\mc K$,
and the subalgebra $\mc V[\partial]$ of differential operators on $\mc V$.

A \emph{rational pseudodifferential operator} with coefficients in $\mc V$
is a pseudodifferential operator $L(\partial)\in\mc V((\partial^{-1}))$
which admits a fractional decomposition
$L(\partial)=A(\partial)B(\partial)^{-1}$,
for some $A(\partial),B(\partial)\in\mc V[\partial]$, $B(\partial)\neq0$.
We denote by $\mc V(\partial)$
the space of all rational pseudodifferential operators with coefficients in $\mc V$.
It is well known that
$\mc K(\partial)$ is the smallest subskewfield of $\mc K((\partial^{-1}))$ containing $\mc V[\partial]$,
see e.g. \cite{CDSK12}.

The algebra of \emph{rational matrix pseudodifferential operators} with coefficients in $\mc V$
is, by definition,  $\Mat_{\ell\times\ell}\mc V(\partial)$.

A matrix differential operator $B(\partial)\in\Mat_{\ell\times\ell}\mc V[\partial]$
is called \emph{non-degenerate}
if it is invertible in $\Mat_{\ell\times\ell}\mc K((\partial^{-1}))$.
Any matrix $M\in\Mat_{\ell\times\ell}\mc V(\partial)$
can be written as a ratio of two matrix differential operators:
$M=A(\partial) B^{-1}(\partial)$,
with $A(\partial),B(\partial)\in\Mat_{\ell\times\ell}\mc V[\partial]$,
and $B(\partial)$ non-degenerate,
see e.g. \cite{CDSK12}.

\subsection{Non-local Poisson structures}\label{sec:3.3}

Let $\mc V$ be an algebra of differential functions in $u_1,\dots,u_\ell$.
Assume that $\mc V$ is a domain, and let $\mc K$ be the corresponding
field of fractions.
To a matrix pseudodifferential operator
$H=\big(H_{ij}(\partial)\big)_{i,j\in I}\in\Mat_{\ell\times\ell}\mc V((\partial^{-1}))$
we associate a map 
$\{\cdot\,_\lambda\,\cdot\}_H:\,\mc V\times\mc V\to\mc V((\lambda^{-1}))$,
given by the following \emph{Master Formula} (cf. \cite{DSK06}):
\begin{equation}\label{20110922:eq1}
\{f_\lambda g\}_H
=
\sum_{\substack{i,j\in I \\ m,n\in\mb Z_+}} 
\frac{\partial g}{\partial u_j^{(n)}}
(\lambda+\partial)^n
H_{ji}(\lambda+\partial)
(-\lambda-\partial)^m
\frac{\partial f}{\partial u_i^{(m)}}
\,\in\mc V((\lambda^{-1}))
\,.
\end{equation}
In particular,
\begin{equation}\label{20130613:eq2}
H_{ji}(\partial)
=
{\{{u_i}_\partial{u_j}\}_H}_\to
\,.
\end{equation}
\begin{theorem}[{\cite[Thm.4.8]{DSK13}}]\label{20110923:prop}
Let $H\in\Mat_{\ell\times\ell}\mc V((\partial^{-1}))$.
Then:
\begin{enumerate}[(a)]
\item
Formula \eqref{20110922:eq1} gives a well-defined non-local $\lambda$-bracket on $\mc V$.
\item
The non-local $\lambda$-bracket $\{\cdot\,_\lambda\,\cdot\}_H$ is skewsymmetric 
if and only if $H$
is a skew-adjoint matrix pseudodifferential operator.
\item
If $H$ is a rational matrix pseudodifferential operator with coefficients in $\mc V$,
then $\{\cdot\,_\lambda\,\cdot\}_H$ is admissible in the sense of equation \eqref{20110921:eq4}.
\item
Let $H$
be a skewadjoint rational matrix pseudodifferential operator with coefficients in $\mc V$.
Then the non-local $\lambda$-bracket $\{\cdot\,_\lambda\,\cdot\}_H$ defined by \eqref{20110922:eq1}
is a Poisson non-local $\lambda$-bracket, i.e. it satisfies the Jacobi identity \eqref{20110922:eq3},
if and only if the Jacobi identity holds on generators ($i,j,k\in I$):
\begin{equation}\label{20110922:eq4}
\{{u_i}_\lambda\{{u_j}_\mu {u_k}\}_H\}_H-\{{u_j}_\mu\{{u_i}_\lambda {u_k}\}_H\}_H
-\{{\{{u_i}_\lambda {u_j}\}_H}_{\lambda+\mu} {u_k}\}_H=0\,,
\end{equation}
where the equality holds in the space $\mc V_{\lambda,\mu}$.
\end{enumerate}
\end{theorem}
\begin{definition}\label{20111007:def}
A \emph{non-local Poisson structure} on $\mc V$
is  a skewadjoint rational matrix pseudodifferential operator $H$ with coefficients in $\mc V$,
satisfying equation \eqref{20110922:eq4} for every $i,j,k\in I$
(which is equivalent to \cite[eq.(6.14)]{DSK13} $H=AB^{-1}$, see Prop.6.11 there).
\end{definition}

\subsection{Hamiltonian equations and integrability}\label{sec:3.4}

Let $\mc V$ be an algebra of differential functions,
which is assumed to be a domain.
We have a non-degenerate pairing 
$(\cdot\,|\,\cdot):\,\mc V^\ell\times\mc V^{\ell}\to\mc V/\partial\mc V$
given by
$(P|\xi)=\tint P\cdot\xi$.
(See e.g. \cite{BDSK09} for a proof of non-degeneracy of this pairing.)
Recall that $\tint$ stands for the canonical projection $\mc V\to\mc V/\partial\mc V$.
Let $H\in\Mat_{\ell\times\ell}\mc V(\partial)$ be a non-local Poisson structure.
We say that $\tint h\in\mc V/\partial\mc V$ and $P\in\mc V^\ell$ are $H$-\emph{associated},
and we denote it by
$$
\tint h\ass{H}P\,,
$$
if there exist
a fractional decomposition $H=AB^{-1}$ 
with $A,B\in\Mat_{\ell\times\ell}\mc V[\partial]$ 
and $B$ non-degenerate,
and an element $F\in\mc K^\ell$ such that
$\frac{\delta h}{\delta u}=BF,\,P=AF$.
Recall that $\frac{\delta h}{\delta u}$ is the vector with coordinates
$\frac{\delta h}{\delta u_i}=\sum_{n\in\mb Z_+}(-\partial)^n\frac{\partial h}{\partial u_i^{(n)}}$.
An evolution equation on the variables $u=\big(u_i\big)_{i\in I}$,
\begin{equation}\label{20120124:eq5}
\frac{du}{dt}
=P\,,
\end{equation}
is called \emph{Hamiltonian} with respect to the Poisson structure $H$
and the Hamiltonian functional $\tint h\in\mc V/\partial\mc V$
if $\tint h\ass{H}P$.

Equation \eqref{20120124:eq5} is called \emph{bi-Hamiltonian}
if there are two compatible non-local Poisson structures $H_0$ and $H_1$,
with fractional decompositions $H_1=AB^{-1}$ and $H_0=CD^{-1}$,
and two local functionals $\tint h_0,\tint h_1\in\mc V/\partial\mc V$,
such that
$\tint h_0\ass{H_1}P$ and $\tint h_1\ass{H_0}P$.

By the chain rule, any element $f\in\mc V$ evolves according to the equation
$$
\frac{df}{dt}=\sum_{i\in I}\sum_{n\in\mb Z_+}(\partial^nP_i)\frac{\partial f}{\partial u_i^{(n)}}
=D_f(\partial)P\,,
$$
and, integrating by parts,
a local functional $\tint f\in\mc V/\partial\mc V$
evolves according to
$$
\frac{d\tint f}{dt}=\int P\cdot\frac{\delta f}{\delta u}
\quad\bigg(=\big(P\big|\frac{\delta f}{\delta u}\big)\bigg)\,.
$$
An \emph{integral of motion} for the Hamiltonian equation \eqref{20120124:eq5}
is a local functional $\tint f\in\mc V/\partial\mc V$
which is constant in time, i.e. such that $(P|\frac{\delta f}{\delta u})=0$.
The usual requirement for \emph{integrability}
is to have infinite linearly independent (over $\mc C$)
sequences $\{\tint h_n\}_{n\in\mb Z_+}\subset\mc V/\partial\mc V$ 
and $\{P_n\}_{n\in\mb Z_+}\subset\mc V^\ell$,
starting with $\tint h_0=\tint h$ and $P_0=P$,
such that
\begin{enumerate}[(i)]
\item
$\frac{\delta h_n}{\delta u}\ass{H}P_n$ for every $n\in\mb Z_+$,
\item
$[P_m,P_n]=0$ for all $m,n\in\mb Z_+$,
\item
$(P_m\,|\,\frac{\delta h_n}{\delta u})=0$ for all $m,n\in\mb Z_+$.
\end{enumerate}
In this case, we have an \emph{integrable hierarchy} of Hamiltonian equations
$$
\frac{du}{dt_n} = P_n\,,\,\,n\in\mb Z_+\,.
$$
Elements $\tint h_n$'s are called \emph{higher Hamiltonians},
the $P_n$'s are called \emph{higher symmetries},
and the condition $(P_m\,|\,\frac{\delta h_n}{\delta u})=0$
says that $\tint h_m$ and $\tint h_n$ are \emph{in involution}.

\section{Quotient algebra of differential functions}\label{sec:4.2a}

Let $\mc V$ be an algebra of differential functions in the variables $u_1,\dots,u_\ell$.
Let $\theta_1,\dots,\theta_m$ be some elements in $\mc V$,
and let $\mc I=\langle\theta_1,\dots,\theta_m\rangle_{\mc V}\subset{\mc V}$
be the differential ideal generated by them.
In general, the quotient differential algebra $\mc V/\mc I$
does not have an induced structure of an algebra of differential functions.
For this, we need, in particular, that the differential ideal $\mc I$ is preserved 
by all partial derivatives $\frac{\partial}{\partial u_i^{(n)}}$,
and this happens only if the elements $\theta_1,\dots,\theta_m$
are of some special form.

The simplest situation is when 
the constraints are some of the differential variables:
$\theta_1=u_{\ell-m+1},\dots,\theta_m=u_\ell$ ($m\leq\ell$).
Since the $\theta_\alpha$'s are in the kernel of  
$\frac{\partial}{\partial u_i^{(n)}}$, for $i=1,\dots,\ell-m$ and $n\in\mb Z_+$,
the differential ideal generated by them
$\mc I=\langle u_{\ell-m+1},\dots,u_{\ell}\rangle_{\mc V}\subset\mc V$
is preserved by all these partial derivatives.
Therefore, 
if $\mc I\cap R_{\ell-m}=0$,
the quotient space $\mc V/\mc I$
is naturally an algebra of differential functions in the variables $u_1,\dots,u_{\ell-m}$.

A more general situation is when
the constraints have the form:
\begin{equation}\label{20130530:eq7}
\theta_\alpha=u_{\ell-m+\alpha}+p_\alpha
\,\,,\,\,\,\,
\alpha=1,\dots,m
\,,
\end{equation}
for some $p_\alpha\in\mc V$,
such that:
\begin{equation}\label{20130530:eq8}
\frac{\partial p_\alpha}{\partial u^{(n)}_{\ell-m+\beta}}=0
\,\,\text{ for all }\,\,
\alpha,\beta=1,\dots,m \text{ and } n\in\mb Z_+
\,.
\end{equation}
In this case we have the following:
\begin{proposition}\label{20130530:prop}
Let $\mc V$ be an algebra of differential functions in the differential variables $u_1,\dots,u_\ell$.
Let $\theta_1,\dots,\theta_m\in\mc V$ ($m\leq\ell$) be elements of the form \eqref{20130530:eq7}.
\begin{enumerate}[(a)]
\item
We have a structure of an algebra of differential functions in the variables $u_1,\dots,u_{\ell-m}$,
which we denote by $\widetilde{\mc V}$,
on the differential algebra $\mc V$,
with the following modified partial derivatives
\begin{equation}\label{20130530:eq9}
\frac{\widetilde{\partial}}{\widetilde{\partial} u_i^{(s)}}
=
\frac{\partial}{\partial u_i^{(s)}}
- \sum_{\alpha=1}^m
\sum_{n=0}^\infty
\frac{\partial (\partial^np_\alpha)}{\partial u_i^{(s)}}
\frac{\partial}{\partial u_{\ell-m+\alpha}^{(n)}}
\,,
\end{equation}
for $i=1,\dots,\ell-m$ and $s\in\mb Z_+$.
\item
All the elements $\theta_\alpha$'s are in the kernel of all partial derivatives
$\frac{\widetilde{\partial}}{\widetilde{\partial} u_i^{(s)}}$.
\item
If, moreover, $\mc I\cap R_{\ell-m}=0$,
we have an induced structure of an algebra of differential functions
on the quotient space $\mc V/\mc I=\widetilde{\mc V}/\mc I$,
with differential variables $u_1,\dots,u_{\ell-m}$,
and with partial derivatives $\frac{\widetilde{\partial}}{\widetilde{\partial} u_i^{(s)}}$,
for $i=1,\dots,\ell-m$ and $s\in\mb Z_+$.
\end{enumerate}
\end{proposition}
\begin{proof}
Note that when we apply $\frac{\widetilde{\partial}}{\widetilde{\partial} u_i^{(s)}}$
to an element $f\in\mc V$,
we only have finitely many non-zero terms,
hence $\frac{\widetilde{\partial}}{\widetilde{\partial} u_i^{(s)}}$ 
is a well-defined derivation of $\mc V$.
Also,
this derivation $\frac{\widetilde{\partial}}{\widetilde{\partial} u_i^{(s)}}$ satisfies the  locality condition
$\frac{\widetilde{\partial} f}{\widetilde{\partial} u_i^{(s)}}=0$
for $s$ sufficiently large.
Indeed,
if $\frac{\partial f}{\partial u_j^{(n)}}=\frac{\partial p_\alpha}{\partial u_j^{(n)}}=0$
for all $j=1,\dots,\ell$ and $n\geq N$,
then 
$\frac{\widetilde{\partial} f}{\widetilde{\partial} u_i^{(s)}}=0$ for all $s>2N$.
Next, all the derivations $\frac{\widetilde{\partial}}{\widetilde{\partial} u_i^{(s)}}$
commute.
Indeed, by the assumption \eqref{20130530:eq8} on the $\theta_\alpha$'s we have
$$
\begin{array}{l}
\displaystyle{
\Big[\frac{\widetilde{\partial}}{\widetilde{\partial} u_i^{(s)}},
\frac{\widetilde{\partial}}{\widetilde{\partial} u_j^{(t)}}\Big]
=
- \sum_{\beta=1}^m
\sum_{r=0}^\infty
\bigg(
\frac{\partial}{\partial u_i^{(s)}}
\frac{\partial (\partial^rp_\beta)}{\partial u_j^{(t)}}
\bigg)
\frac{\partial}{\partial u_{\ell-m+\beta}^{(r)}}
} \\
\displaystyle{
+ \sum_{\alpha=1}^m
\sum_{n=0}^\infty
\bigg(
\frac{\partial}{\partial u_j^{(t)}}
\frac{\partial (\partial^np_\alpha)}{\partial u_i^{(s)}}
\bigg)
\frac{\partial}{\partial u_{\ell-m+\alpha}^{(n)}}
=0\,,}
\end{array}
$$
since $\frac{\partial}{\partial u_i^{(s)}}$ and $\frac{\partial}{\partial u_j^{(t)}}$ commute.
In order to complete the proof of (a),
we are left to prove that the derivations $\frac{\widetilde{\partial}}{\widetilde{\partial} u_i^{(s)}}$
satisfy the commutation rules \eqref{eq:0.4}.
We have
$$
\begin{array}{l}
\displaystyle{
\Big[\frac{\widetilde{\partial}}{\widetilde{\partial} u_i^{(s)}},\partial\Big]
=
\Big[
\frac{\partial}{\partial u_i^{(s)}}
,\partial\Big]
- \sum_{\alpha=1}^m
\sum_{n=0}^\infty
\Big[
\frac{\partial (\partial^np_\alpha)}{\partial u_i^{(s)}}
\frac{\partial}{\partial u_{\ell-m+\alpha}^{(n)}}
,\partial\Big]
} \\
\displaystyle{
=
\frac{\partial}{\partial u_i^{(s-1)}}
- \sum_{\alpha=1}^m
\sum_{n=0}^\infty
\frac{\partial (\partial^np_\alpha)}{\partial u_i^{(s)}}
\Big[
\frac{\partial}{\partial u_{\ell-m+\alpha}^{(n)}}
,\partial\Big]
+ \sum_{\alpha=1}^m
\sum_{n=0}^\infty
\Big(\partial
\frac{\partial (\partial^np_\alpha)}{\partial u_i^{(s)}}
\Big)
\frac{\partial}{\partial u_{\ell-m+\alpha}^{(n)}}
} \\
\displaystyle{
=
\frac{\partial}{\partial u_i^{(s-1)}}
- \sum_{\alpha=1}^m
\sum_{n=0}^\infty
\frac{\partial (\partial^{n+1}p_\alpha)}{\partial u_i^{(s)}}
\frac{\partial}{\partial u_{\ell-m+\alpha}^{(n)}}
+ \sum_{\alpha=1}^m
\sum_{n=0}^\infty
\Big(\partial
\frac{\partial (\partial^np_\alpha)}{\partial u_i^{(s)}}
\Big)
\frac{\partial}{\partial u_{\ell-m+\alpha}^{(n)}}
} \\
\displaystyle{
=
\frac{\partial}{\partial u_i^{(s-1)}}
- \sum_{\alpha=1}^m
\sum_{n=0}^\infty
\frac{\partial (\partial^{n}p_\alpha)}{\partial u_i^{(s-1)}}
\frac{\partial}{\partial u_{\ell-m+\alpha}^{(n)}}
=
\frac{\widetilde{\partial}}{\widetilde{\partial} u_i^{(s-1)}}
\,.}
\end{array}
$$
In the fourth equality we used the identity
$$
\frac{\partial (\partial^{n+1}p_\alpha)}{\partial u_i^{(s)}}
=
\partial\frac{\partial (\partial^{n}p_\alpha)}{\partial u_i^{(s)}}
+
\frac{\partial (\partial^{n}p_\alpha)}{\partial u_i^{(s-1)}}
\,,
$$
which holds due to \eqref{eq:0.4}.

Next, let us prove part (b).
Since, by assumption \eqref{20130530:eq8},
$p_\alpha$ is independent of the variables $u_{\ell-m+1},\dots,u_\ell$,
we have, for $i=1,\dots,\ell-m$ and $s\in\mb Z_+$,
$$
\begin{array}{l}
\displaystyle{
\frac{\widetilde{\partial}\theta_\alpha}{\widetilde{\partial} u_i^{(s)}}
=
\Big(\frac{\partial}{\partial u_i^{(s)}}
- \sum_{\beta=1}^m
\sum_{n=0}^\infty
\frac{\partial (\partial^np_\beta)}{\partial u_i^{(s)}}
\frac{\partial}{\partial u_{\ell-m+\beta}^{(n)}}
\Big)(u_{\ell-m+\alpha}+p_\alpha)
} \\
\displaystyle{
=
\frac{\partial p_\alpha}{\partial u_i^{(s)}}
- \sum_{\beta=1}^m
\sum_{n=0}^\infty
\frac{\partial (\partial^np_\beta)}{\partial u_i^{(s)}}
\frac{\partial u_{\ell-m+\alpha}}{\partial u_{\ell-m+\beta}^{(n)}}
=0
\,.}
\end{array}
$$

Finally, since, by assumption, $\mc I\cap R_{\ell-m}=0$,
we have a natural embedding $R_{\ell-m}\subset\mc V/\mc I$.
Hence, part (c) is an immediate consequence of part (b).
\end{proof}

Note that the Frechet derivative \eqref{20120126:eq2}
of a collection of elements as in \eqref{20130530:eq7} has the form
\begin{equation}\label{20130604:eq4}
D_{\theta}(\partial)=(D_p(\partial)\,\,\id_{m})
\,,
\end{equation}
where
\begin{equation}\label{20130604:eq2}
D_p(\partial)
=\Big(
\sum_{n\in\mb Z_+}\frac{\partial p_\alpha}{u_i^{(n)}}\partial^n
\Big)_{\substack{\alpha=1,\dots,m \\ i=1,\dots,\ell-m}}
\in\Mat_{m\times(\ell-m)}\mc V[\partial]
\,,
\end{equation}
and $\id_m$ is the $m\times m$ identity matrix.

We can find the formula for the variational derivatives
in the algebra of differential functions $\widetilde{\mc V}$ 
(with the modified partial derivatives \eqref{20130530:eq9}).
Namely, for $i=1,\dots,\ell-m$, we have
$$
\frac{\widetilde{\delta}}{\widetilde{\delta} u_i}
:=
\sum_{n\in\mb Z_+}(-\partial)^n\frac{\widetilde{\partial}}{\widetilde{\partial} u_i^{(n)}}
=
\frac{\delta}{\delta u_i}
-
\sum_{\alpha=1}^m
\sum_{s\in\mb Z_+}
(-\partial)^s
\frac{\partial p_\alpha}{\partial u_i^{(s)}}
\frac{\delta}{\delta u_{\ell-m+\alpha}}
\,.
$$
In the last identity we used the commutation relation \eqref{eq:0.4b}.
We can rewrite the above equation in matrix form as follows:
\begin{equation}\label{20130603:eq4}
\frac{\widetilde{\delta}}{\widetilde{\delta} u}
=
\big(\id_{\ell-m}\,,\,-D_p^*(\partial)\big)
\circ
\frac{\delta}{\delta u}
\,.
\end{equation}
The variational derivatives $\frac{\widetilde{\delta}}{\widetilde{\delta} u_i}$
on the quotient algebra $\mc V/\mc I$ are induced by \eqref{20130603:eq4}:
$$
\frac{\widetilde{\delta}\overline{f}}{\widetilde{\delta} u_i}
=
\overline{
\frac{\widetilde{\delta}f}{\widetilde{\delta} u_i}
}
\,\,,\,\,\,\,
i=1,\dots,\ell-m
\,.
$$
Here and further, for $f\in\mc V$, we let $\bar{f}$ be its coset in the quotient space $\mc V/\mc I$.

\section{Central reduction for (non-local) Poisson structures and Hamiltonian equations}\label{sec:4a}

\subsection{Central elements and constant densities}\label{sec:4.2}

\begin{definition}\label{20130530:def}
Let $\mc V$ be an algebra of differential functions, which is a domain.
Let $H\in\Mat_{\ell\times\ell}\mc V((\partial^{-1}))$ be a (non-local) Poisson structure 
on $\mc V$.
\begin{enumerate}[(i)]
\item
An element $\theta\in\mc V$ is called \emph{central} for $H$ if
$D_\theta(\partial)\circ H(\partial)=0$.
\item
An element $\theta\in\mc V$ is called a \emph{constant density}
for the evolution equation $\frac{du}{dt}=P\,\in\mc V^\ell$
if $\frac{d\theta}{dt}\Big(=D_\theta(\partial)P\Big)=0$.
\item
An element $\tint\theta\in\mc V/\partial\mc V$ is called an \emph{integral of motion}
(and $\theta\in\mc V$ is the corresponding \emph{conserved density})
for the evolution equation $\frac{du}{dt}=P\,\in\mc V^\ell$
if $\frac{d}{dt}\tint \theta\Big(=\tint \frac{\delta\theta}{\delta u}\cdot P\Big)=0$.
\end{enumerate}
\end{definition}
\begin{lemma}\phantomsection\label{20130530:lem}
\begin{enumerate}[(a)]
\item
An element $\theta\in\mc V$ is central for the Poisson structure $H(\partial)$
if and only if it is a central element for the 
corresponding PVA $\lambda$-bracket $\{\cdot\,_\lambda\,\cdot\}_H$
given by the Master Formula \eqref{20110922:eq1}.
\item
If $\theta\in\mc V$ is a central element for the Poisson structure $H$,
then it is a constant density for every evolution equation
which is Hamiltonian with respect to the Poisson structure $H$.
\item
If $\theta\in\mc V$ is a constant density for the evolution equation $\frac{du}{dt}=P$,
then $\tint\theta\in\mc V/\partial\mc V$ is an integral of motion for the same evolution equation.
\end{enumerate}
\end{lemma}
\begin{proof}
By the Master Formula \eqref{20110922:eq1} we have
$$
{\{f_\partial\theta\}_H}_\to=
D_\theta(\partial)\circ H(\partial)\circ D_f^*(\partial)
\,.
$$
Hence, $\theta$ is central for the PVA $\lambda$-bracket $\{\cdot\,_\lambda\,\cdot\}_H$
if and only if the RHS of this equation is zero for every $f\in\mc V$,
which is equivalent to the equation $D_\theta(\partial)\circ H(\partial)=0$.
This proves part (a).


Next, let us prove part (b).
Recall from Section \ref{sec:3.4}
that the evolution equation $\frac{du}{dt}=P$ is Hamiltonian
with respect to the Poisson structure $H$ and the Hamiltonian functional $\tint h$,
if $\tint h\ass{H}P$,
i.e. if there exist a fractional decomposition $H(\partial)=A(\partial)\circ B^{-1}(\partial)$,
with $A(\partial),B(\partial)\in\Mat_{\ell\times\ell}\mc V[\partial]$
and $B(\partial)$ non-degenerate,
and an element $F\in\mc K^\ell$,
such that $P=A(\partial)F$ and $\frac{\delta h}{\delta u}=B(\partial)F$.
By assumption, $\theta$ is central for $H$,
i.e. $D_\theta(\partial)\circ H(\partial)=0$,
and, since $B(\partial)$ is non-degenerate,
it follows that $D_\theta(\partial)\circ A(\partial)=0$.
But then $D_\theta(\partial)P=D_\theta(\partial)A(\partial)F=0$.
Therefore, $\theta$ is a constant density for the given Hamiltonian equation, as claimed.

Part (c) is immediate. Indeed, applying integral to the condition $D_\theta(\partial)P=0$
and integrating by parts, we get $\tint\frac{\delta\theta}{\delta u}\cdot P=0$.
\end{proof}
\begin{remark}\label{20130605:rem}
Recall that a local functional $\tint h\in\mc V/\partial\mc V$
is called a \emph{Casimir element} 
for the local Poisson structure $H(\partial)\in\Mat_{\ell\times\ell}\mc V[\partial]$
if $H(\partial)\frac{\delta h}{\delta u}=0$.
Since $\frac{\delta h}{\delta u}=D_h^*(\partial)(\id_\ell)$,
it is immediate to check that every central element for $H$ is a Casimir element.
More generally,
a local functional $\tint h\in\mc V/\partial\mc V$
is called a \emph{Casimir element} for 
a non-local Poisson structure $H$ if $\tint h\ass{H}0$.
It is easy to check, using Lemma \ref{20130611:lem} below,
that if $\tint h$ is a central element for $H$,
then it is still a Casimir element.
\end{remark}

\subsection{Central reduction of a Poisson structure}\label{sec:4.3a}

Let $H(\partial)\in\Mat_{\ell\times\ell}\mc V((\partial^{-1}))$ 
be a (non-local) Poisson structure on the algebra of differential functions $\mc V$,
and let $\theta_1,\dots,\theta_m\in\mc V$ be central elements for $H(\partial)$.
Let $\mc I=\langle\theta_1,\dots,\theta_m\rangle_{\mc V}\subset\mc V$
be the differential ideal generated by $\theta_1,\dots,\theta_m$.
Due to Lemma \ref{20130530:lem}(a),
$\mc I$ is an ideal for the PVA $\lambda$-bracket $\{\cdot\,_\lambda\,\cdot\}_H$
associated to the Poisson structure $H(\partial)$,
and therefore we have an induced PVA structure on $\mc V/\mc I$.
The corresponding PVA $\lambda$-bracket on $\mc V/\mc I$
is given by
\begin{equation}\label{20130530:eq5b}
\{\bar{f}_\lambda\bar{g}\}_H^{\overline{\phantom{II}}}
=
\sum_{\substack{i,j\in I \\ m,n\in\mb Z_+}} 
\overline{\frac{\partial g}{\partial u_j^{(n)}}}
(\lambda+\partial)^n
\overline{
H_{ji}}(\lambda+\partial)
(-\lambda-\partial)^m
\overline{\frac{\partial f}{\partial u_i^{(m)}}}
\,.
\end{equation}
As before, for $f\in\mc V$, we let $\bar{f}$ be its coset in the quotient space $\mc V/\mc I$,
and also, for $h(\partial)\in\mc V((\partial^{-1}))$,
we denote by $\overline{h}(\partial)\in(\mc V/\mc I)((\partial^{-1}))$
the pseudodifferential operator obtained by taking 
the cosets of all coefficients of $h(\partial)$.
We would like to prove that, in fact,
this PVA $\lambda$-bracket is associated to a (non-local) Poisson structure
on $\mc V/\mc I$.
For this we need, in particular, that the quotient differential algebra $\mc V/\mc I$
is an algebra of differential functions.
By Proposition \ref{20130530:prop},
this happens, for example, if the elements $\theta_\alpha$'s
are of the special form \eqref{20130530:eq7}, with $\mc I\cap R_{\ell-m}=0$,
and in this case we need to take the modified partial derivatives \eqref{20130530:eq9}.
We want to prove that, in this case,
we indeed have an induced Poisson structure $H^C(\partial)$ on $\mc V/\mc I$,
which we call the \emph{central reduction} of $H(\partial)$.
\begin{proposition}\label{20130530:propb}
Let $\mc V$ be an algebra of differential functions in the differential variables $u_1,\dots,u_\ell$,
which is a domain.
Let $\theta_1,\dots,\theta_m\in\mc V$ be as in \eqref{20130530:eq7},
and let $\mc I=\langle\theta_1,\dots,\theta_m\rangle_\mc V\subset\mc V$
be the differential ideal generated by them.
Let $H(\partial)\in\Mat_{\ell\times\ell}\mc V((\partial^{-1}))$ be a Poisson structure on $\mc V$.
\begin{enumerate}[(a)]
\item
The elements $\theta_1,\dots,\theta_m$ are central for $H(\partial)$
if and only if the matrix $H(\partial)$ has the following form
\begin{equation}\label{20130604:eq1}
H(\partial)=
\left(\begin{array}{c}
\id_{\ell-m} \\
-D_p(\partial)
\end{array}
\right)
\circ A(\partial)\circ
\big(\id_{\ell-m}\,,\,\,-D_p^*(\partial)\big)
\,,
\end{equation}
where $D_p(\partial)\in\Mat_{m\times(\ell-m)}\mc V[\partial]$
is the matrix \eqref{20130604:eq2},
and $A(\partial)$
is a rational $(\ell-m)\times(\ell-m)$ matrix pseudodifferential operator with coefficients in $\mc V$.
\item
If $H(\partial)$ has the form \eqref{20130604:eq1},
then
\begin{equation}\label{20130604:eq3}
\begin{array}{rcl}
\displaystyle{
\{f_\lambda g\}_H^{\mc V}
}
&=&
\displaystyle{
\sum_{i,j=1}^\ell \sum_{m,n=0}^\infty
\frac{\partial g}{\partial u_j^{(n)}}
(\lambda+\partial)^n
H_{ji}(\lambda+\partial)
(-\lambda-\partial)^m
\frac{\partial f}{\partial u_i^{(m)}}
} \\
&=&
\displaystyle{
\sum_{i,j=1}^{\ell-m} \sum_{s,t=0}^\infty
\frac{\widetilde{\partial} g}{\widetilde{\partial} u_j^{(t)}}
(\lambda+\partial)^t
A_{ts}(\lambda+\partial)
(-\lambda-\partial)^s
\frac{\widetilde{\partial} f}{\widetilde{\partial} u_i^{(m)}}
=\{f_\lambda g\}_{A}^{\widetilde{\mc V}}
\,.}
\end{array}
\end{equation}
In other words,
the matrix $A(\partial)$ is a Poisson structure 
on the algebra of differential functions $\widetilde{\mc V}$
defined in Proposition \ref{20130530:prop}(a),
and the $H$-$\lambda$-bracket on $\mc V$
(with usual partial derivatives)
coincides with 
the $A$-$\lambda$-bracket on $\widetilde{\mc V}$
($=\mc V$ with modified partial derivatives).
\item
Assume that $\mc I\cap R_{\ell-m}=0$,
that the quotient algebra of differential function $\mc V/\mc I$
(cf. Proposition \ref{20130530:prop}) is a domain,
and, for $H(\partial)$ of the form \eqref{20130604:eq1},
assume that the induced matrix 
$\overline{A}(\partial)\in\Mat_{(\ell-m)\times(\ell-m)}(\mc V/\mc I)((\partial^{-1}))$
is rational.
Then, we have a well defined \emph{central reduction} of $H$ 
by the central elements $\theta_\alpha$'s,
given by 
$$
H^C(\partial)=\overline{A}(\partial)\,.
$$
In other words,
the induced PVA $\lambda$-bracket \eqref{20130530:eq5b} on $\mc V/\mc I$
is associated to the \emph{centrally reduced} Poisson structure $H^C(\partial)$ on $\mc V/\mc I$.
\end{enumerate}
\end{proposition}
\begin{proof}
Write the matrix $H(\partial)$
in block form as follows:
\begin{equation}\label{20130530:eq2}
H(\partial)=
\left(\begin{array}{cc}
A(\partial) & B(\partial) \\
-B^*(\partial) & D(\partial)
\end{array}
\right)
\,,
\end{equation}
where 
$A(\partial)$ is an $(\ell-m)\times(\ell-m)$ matrix,
$B(\partial)$ is an $(\ell-m)\times m$ matrix,
and $D(\partial)$ is an $m\times m$ matrix.
Recalling equation \eqref{20130604:eq4},
we have that, by definition, the elements $\theta_1,\dots,\theta_m$
are central for $H$ if and only if
$$
D_\theta(\partial)\circ H(\partial)
=
(D_p(\partial)\,\,\id_{m})
\circ
\left(\begin{array}{cc}
A(\partial) & B(\partial) \\
-B^*(\partial) & D(\partial)
\end{array}
\right)
=0\,,
$$
namely $B^*(\partial)=D_p(\partial)\circ A(\partial)$
(which is the same as $B(\partial)=-A(\partial)\circ D_p^*(\partial)$,
since $A(\partial)$ is skewadjoint),
and $D(\partial)=-D_p(\partial)\circ B(\partial)=D_p(\partial)\circ A(\partial)\circ D_p^*(\partial)$.
Part (a) follows.

Next, we prove part (b).
By the Master Formula \eqref{20110922:eq1}
and the block form \eqref{20130604:eq1} for $H$, 
we have
\begin{equation}\label{20130603:eq2}
\begin{array}{l}
\displaystyle{
\{f_\lambda g\}_H
=
\sum_{i,j=1}^\ell\sum_{s,t\in\mb Z_+}
\frac{\partial g}{\partial u_j^{(t)}}
(\lambda+\partial)^t
H_{ji}(\lambda+\partial)
(-\lambda-\partial)^s
\frac{\partial f}{\partial u_i^{(s)}}
} \\
\displaystyle{
=
\sum_{i,j=1}^{\ell-m}\sum_{s,t\in\mb Z_+}
\frac{\partial g}{\partial u_j^{(t)}}
(\lambda+\partial)^t
A_{ji}(\lambda+\partial)
(-\lambda-\partial)^s
\frac{\partial f}{\partial u_i^{(s)}}
} \\
\displaystyle{
-
\sum_{i,j=1}^{\ell-m}\sum_{\beta=1}^m\sum_{s,q\in\mb Z_+}
\frac{\partial g}{\partial u_{\ell-m+\beta}^{(q)}}
(\lambda+\partial)^q
(D_p)_{\beta j}(\lambda+\partial)
A_{ji}(\lambda+\partial)
(-\lambda-\partial)^s
\frac{\partial f}{\partial u_i^{(s)}}
} \\
\displaystyle{
-
\sum_{i,j=1}^{\ell-m}\sum_{\alpha=1}^m\sum_{p,t\in\mb Z_+}
\frac{\partial g}{\partial u_j^{(t)}}
(\lambda+\partial)^t
A_{ji}(\lambda+\partial) (D_p^*)_{i\alpha}(\lambda+\partial)
(-\lambda-\partial)^p
\frac{\partial f}{\partial u_{\ell-m+\alpha}^{(p)}}
} \\
\displaystyle{
+
\sum_{i,j=1}^{\ell-m}\sum_{\alpha,\beta=1}^m\sum_{p,q\in\mb Z_+}
\frac{\partial g}{\partial u_{\ell-m+\beta}^{(q)}}
(\lambda+\partial)^q
(D_p)_{\beta j}(\lambda+\partial)A_{ji}(\lambda+\partial)
} \\
\displaystyle{
\,\,\,\,\,\,\,\,\,\,\,\,\,\,\,\,\,\,\,\,\,\,\,\,\,\,\,\,\,\,\,\,\,\,\,\,\,\,\,\,\,\,\,\,\,\,\,\,\,\,\,\,\,\,
\,\,\,\,\,\,\,\,\,\,\,\,\,\,\,\,\,\,\,\,\,\,\,\,\,\,\,
\times
(D_p^*)_{i\alpha}(\lambda+\partial)
(-\lambda-\partial)^p
\frac{\partial f}{\partial u_{\ell-m+\alpha}^{(p)}}
\,.}
\end{array}
\end{equation}
By the definition of Frechet derivative and the commutation relation \eqref{eq:0.4b},
we have
$$
\begin{array}{l}
\displaystyle{
\frac{\partial g}{\partial u_{\ell-m+\beta}^{(q)}}
(\lambda+\partial)^q
(D_p)_{\beta j}(\lambda+\partial)
=
\sum_{t\in\mb Z_+}
\frac{\partial(\partial^qp_\beta)}{\partial u_j^{(t)}}
\frac{\partial g}{\partial u_{\ell-m+\beta}^{(q)}}
(\lambda+\partial)^t
\,,} \\
\displaystyle{
(D_p^*)_{i\alpha}(\lambda+\partial)
(-\lambda-\partial)^p
\frac{\partial f}{\partial u_{\ell-m+\alpha}^{(p)}}
=
\sum_{s\in\mb Z_+}
(-\lambda-\partial)^s
\frac{\partial(\partial^pp_\alpha)}{\partial u_i^{(s)}}
\frac{\partial f}{\partial u_{\ell-m+\alpha}^{(p)}}
\,.}
\end{array}
$$
We can thus rewrite equation \eqref{20130603:eq2} as
$$
\begin{array}{l}
\displaystyle{
\{f_\lambda g\}_H
=
\sum_{i,j=1}^{\ell-m}\sum_{s,t\in\mb Z_+}
\frac{\partial g}{\partial u_j^{(t)}}
(\lambda+\partial)^t
A_{ji}(\lambda+\partial)
(-\lambda-\partial)^s
\frac{\partial f}{\partial u_i^{(s)}}
} \\
\displaystyle{
-
\sum_{i,j=1}^{\ell-m}\sum_{\beta=1}^m\sum_{s,t,q\in\mb Z_+}
\frac{\partial(\partial^qp_\beta)}{\partial u_j^{(t)}}
\frac{\partial g}{\partial u_{\ell-m+\beta}^{(q)}}
(\lambda+\partial)^t
A_{ji}(\lambda+\partial)
(-\lambda-\partial)^s
\frac{\partial f}{\partial u_i^{(s)}}
} \\
\displaystyle{
-
\sum_{i,j=1}^{\ell-m}\sum_{\alpha=1}^m\sum_{s,t,p\in\mb Z_+}
\frac{\partial g}{\partial u_j^{(t)}}
(\lambda+\partial)^t
A_{ji}(\lambda+\partial) 
(-\lambda-\partial)^s
\frac{\partial(\partial^pp_\alpha)}{\partial u_i^{(s)}}
\frac{\partial f}{\partial u_{\ell-m+\alpha}^{(p)}}
} \\
\displaystyle{
+
\sum_{i,j=1}^{\ell-m}\sum_{\alpha,\beta=1}^m\sum_{s,t,p,q\in\mb Z_+}
\frac{\partial(\partial^qp_\beta)}{\partial u_j^{(t)}}
\frac{\partial g}{\partial u_{\ell-m+\beta}^{(q)}}
(\lambda+\partial)^t
A_{ji}(\lambda+\partial)
} \\
\displaystyle{
\,\,\,\,\,\,\,\,\,\,\,\,\,\,\,\,\,\,\,\,\,\,\,\,\,\,\,\,\,\,\,\,\,\,\,\,\,\,\,\,\,\,\,\,\,\,\,\,\,\,\,\,\,\,
\,\,\,\,\,\,\,\,\,\,\,\,\,\,\,\,\,\,\,\,\,\,\,\,\,\,\,
\times
(-\lambda-\partial)^s
\frac{\partial(\partial^pp_\alpha)}{\partial u_i^{(s)}}
\frac{\partial f}{\partial u_{\ell-m+\alpha}^{(p)}}
} \\
\displaystyle{
=
\sum_{i,j=1}^{\ell-m}\sum_{s,t\in\mb Z_+}
\frac{\widetilde{\partial} g}{\widetilde{\partial} u_j^{(t)}}
(\lambda+\partial)^t
A_{ji}(\lambda+\partial)
(-\lambda-\partial)^s
\frac{\widetilde{\partial} f}{\widetilde{\partial} u_i^{(s)}}
\,.}
\end{array}
$$
Hence, equation \eqref{20130604:eq3} holds.
Part (c) is an immediate consequence of part (b).
\end{proof}

\begin{remark}\label{20130622:rem1}
We believe that the assumption that
$\overline{A}(\partial)$ is rational in part (c) of Theorem \ref{20130530:propb} is 
automatically satisfied.
First, note that we can prove it entrywise, so we reduce to the scalar case. 
Then, we would need to prove the following.
Let $\mc V$ be a differential domain, and
let $a(\partial),b(\partial)\in\mc V[\partial]$, with $b(\partial)\neq0$,
be such that $h(\partial)=a(\partial)\circ b^{-1}(\partial)\in\mc V((\partial^{-1}))$.
Then, there should exist
$f\in\mc V$ such that $a(\partial)\circ\frac1f, b(\partial)\circ\frac1f\in\mc V[\partial]$,
and $fb^{-1}(\partial)\in\mc V((\partial^{-1}))$.
\end{remark}

\subsection{Central reduction of a Hamiltonian equation}\label{sec:4.4a}

Let $\mc V$ be an algebra of differential functions, which is a domain,
and let $H(\partial)\in\Mat_{\ell\times\ell}\mc V((\partial^{-1}))$
be a (non-local) Poisson structure on $\mc V$.
Let
\begin{equation}\label{20130530:eq1}
\frac{du}{dt}=P\,\in\mc V^\ell\,,
\end{equation}
be a Hamiltonian equation associated to the Poisson structure $H$
and to a Hamiltonian functional $\tint h\in\mc V/\partial\mc V$.
In other words, $\tint h\ass{H}P$ (cf. Section \ref{sec:3.4}).

Suppose that $\theta_1,\dots,\theta_m\in\mc V$
are central elements for $H$, of the form \eqref{20130530:eq7}.
By Proposition \ref{20130530:prop}(c),
if $\mc I\cap R_{\ell-m}=0$,
the quotient space $\mc V/\mc I$ has a natural structure of
algebra of differential functions in the variables $u_1,\dots,u_{\ell-m}$,
with modified partial derivatives $\frac{\widetilde{\partial}}{\widetilde{\partial} u_i^{(s)}}$,
and in Proposition \ref{20130530:propb}(c)
we define a central reduction $H^C(\partial)$,
which is a Poisson structure on $\mc V/\mc I$.
The following proposition
says that equation \eqref{20130530:eq1}
can be ``reduced'' to an evolution equation in $\mc V/\mc I$,
which is Hamiltonian for the Poisson structure $H^C(\partial)$.

\begin{proposition}\label{20130604:prop}
Let $\mc V$ be an algebra of differential functions in $u_1,\dots,u_\ell$,
which is a domain.
Let $\theta_1,\dots,\theta_m\in\mc V$ be elements of the form \eqref{20130530:eq7},
and let $H(\partial)$ be a Poisson structure on $\mc V$ of the form \eqref{20130604:eq1}.
(In particular, by Proposition \ref{20130530:propb}(a), $\theta_1,\dots,\theta_m$ are central for $H$.)
Let $\tint h\in\mc V/\partial\mc V$ and $P\in\mc V^\ell$ 
be $H$-associated over $\mc V$, $\tint h\ass{H}P$,
so that equation \eqref{20130530:eq1}
is Hamiltonian with respect to the Poisson structure $H$ and the Hamiltonian functional $\tint h$.
Let $P_1\in\mc V^{\ell-m}$ be obtained by taking the first $\ell-m$ entries of $P\in\mc V^\ell$.
Then:
\begin{enumerate}[(a)]
\item
The elements $\tint h\in\mc V/\partial\mc V$ and $P_1\in\mc V^{\ell-m}$ 
are $A$-associated over $\widetilde{\mc V}$: $\tint h\ass{A}P_1$.
Hence, equation \eqref{20130530:eq1}
can also be viewed as a  Hamiltonian equation
on the modified algebra of differential functions $\widetilde{\mc V}$,
for the Poisson structure $A(\partial)$ and the same Hamiltonian functional $\tint h$.
\item
Assume that $\mc I\cap R_{\ell-m}=0$,
and that the quotient algebra of differential functions $\mc V/\mc I$
(cf. Proposition \ref{20130530:prop}) is a domain.
Assume moreover that,
for $H(\partial)$ of the form \eqref{20130604:eq1},
the rational matrix $A(\partial)$ has a fractional decomposition
$A(\partial)=M_1(\partial)\circ N_1(\partial)^{-1}$,
with $M_1(\partial),N_1(\partial)\in\Mat_{(\ell-m)\times(\ell-m)}\mc V[\partial]$
and $N_1(\partial)$ non-degenerate,
such that $\overline{N_1}(\partial)\in\Mat_{(\ell-m)\times(\ell-m)}(\mc V/\mc I)[\partial]$
is also non-degenerate.
Hence, $\overline{A}(\partial)$ is rational,
and by Proposition \ref{20130530:propb}(c),
we have a centrally reduced 
Poisson structure $H^C(\partial)=\overline{A}(\partial)$ on $\mc V/\mc I$.
Then,
$\tint \overline{h}\ass{H^C}\overline{P_1}$.
In particular, the following evolution equation in $\mc V/\mc I$,
$$
\frac{du}{dt}=\overline{P_1}\,\in(\mc V/\mc I)^\ell\,,
$$
is a  Hamiltonian equation
on the quotient algebra $\mc V/\mc I$,
for the centrally reduced Poisson structure $H^C(\partial)$,
and the Hamiltonian functional $\tint \overline{h}$.
\end{enumerate}
\end{proposition}
\begin{proof}
By assumption we have the association relation $\tint h\ass{H}P$.
This means that we have a fractional decomposition 
$H(\partial)=M(\partial)N^{-1}(\partial)$,
with $M(\partial),N(\partial)\in\Mat_{\ell\times\ell}\mc V[\partial]$
and $N(\partial)$ non-degenerate,
and an element $F\in\mc K^\ell$,
such that
\begin{equation}\label{20130607:eq1}
\frac{\delta h}{\delta u}=N(\partial)F
\,\,,\,\,\,\,
P=M(\partial)F\,.
\end{equation}


The first observation is that,
without loss of generality,
we can assume that $N(\partial)$ is upper triangular.
Indeed, by \cite[Lem.3.1]{CDSK13b},
we have $N(\partial)=T(\partial)\circ \frac1fU(\partial)$,
where $T(\partial)\in\Mat_{\ell\times\ell}\mc V[\partial]$ is upper triangular non-degenerate,
$f\in\mc V$ is a non-zero element,
and $U(\partial)\in\Mat_{\ell\times\ell}\mc V[\partial]$
is invertible in $\Mat_{\ell\times\ell}\mc K[\partial]$.
Clearing the denominators, 
we can write $U^{-1}(\partial)\circ g\in\Mat_{\ell\times\ell}\mc V[\partial]$,
for some $g\in\mc V$.
Then $H(\partial)$ admits the fractional decomposition
$H(\partial)=\tilde{M}(\partial)\tilde{N}^{-1}(\partial)$,
where
$$
\tilde{M}(\partial)=
M(\partial)\circ U^{-1}(\partial)\circ fg
\,\,,\,\,\,\,
\tilde{N}(\partial)=
N(\partial)\circ U^{-1}(\partial)\circ fg
=T(\partial)\circ g
\,,
$$
and equations \eqref{20130607:eq1} imply
$$
\frac{\delta h}{\delta u}
=
\tilde{N}(\partial)\tilde{F}
\,\,,\,\,\,\,
P
=
\tilde{M}(\partial)\tilde{F}
\,,
$$
where $\tilde{F}=\frac1{fg}U(\partial)F\in\mc K^\ell$.
So, we can replace the original fractional decomposition $H=MN^{-1}$
by the new fractional decomposition $H=\tilde{M}\tilde{N}^{-1}$,
where $\tilde{N}$ is upper triangular.


Since, by assumption, $H(\partial)$ has the form \eqref{20130604:eq1},
we have $(D_p(\partial)\,,\,\id_m)\circ H(\partial)=0$,
and therefore $(D_p(\partial)\,,\,\id_{\ell-m})\circ M(\partial)=0$.
Hence, $M(\partial)$ has the following form:
\begin{equation}\label{20130607:eq2}
M(\partial)=
\left(\begin{array}{c}
\id_{\ell-m} \\ -D_p(\partial)
\end{array}
\right)
\circ
\big(
M_1(\partial) \, M_2(\partial)
\big)
\,,
\end{equation}
where $M_1(\partial)\in\Mat_{(\ell-m)\times(\ell-m)}\mc V[\partial]$
and $M_2(\partial)\in\Mat_{(\ell-m)\times m}\mc V[\partial]$.
Also, we let
$$
N(\partial)=
\left(\begin{array}{cc}
N_1(\partial) & N_2(\partial) \\
0 & N_4(\partial)
\end{array}
\right)
\,,
$$
where 
$N_1(\partial)\in\Mat_{(\ell-m)\times(\ell-m)}\mc V[\partial]$,
$N_2(\partial)\in\Mat_{(\ell-m)\times m}\mc V[\partial]$,
and $N_4(\partial)\in\Mat_{m\times m}\mc V[\partial]$,
with $N_1(\partial)$ and $N_4(\partial)$ non-degenerate.
Its inverse is
\begin{equation}\label{20130607:eq4}
N^{-1}(\partial)=
\left(\begin{array}{cc}
N_1^{-1}(\partial) & -N_1^{-1}(\partial)\circ N_2(\partial)\circ N_4^{-1}(\partial) \\
0 & N_4^{-1}(\partial)
\end{array}
\right)
\,\in\Mat_{\ell\times\ell}\mc V((\partial^{-1}))
\,.
\end{equation}

It is easy to deduce that
the matrix $A(\partial)\in\Mat_{(\ell-m)\times(\ell-m)}\mc V((\partial^{-1}))$
in \eqref{20130604:eq1} admits the fractional decomposition
\begin{equation}\label{20130607:eq7}
A(\partial)=
M_1(\partial)\circ N_1^{-1}(\partial)
\,.
\end{equation}

Since, by assumption, $H(\partial)$ has the form \eqref{20130604:eq1},
we have 
$$
M(\partial)\circ N^{-1}(\partial)\circ
\left(\begin{array}{c}
D_p^*(\partial) \\ \id_m
\end{array}
\right)=0
\,.
$$
Substituting the expressions \eqref{20130607:eq2} and \eqref{20130607:eq4} of $M$ and $N^{-1}$
in this equation, we get
$$
M_1(\partial)\circ N_1^{-1}(\partial)\circ D_p^*(\partial)
-M_1(\partial)\circ N_1^{-1}(\partial)\circ N_2(\partial)\circ N_4^{-1}(\partial)
+M_2(\partial)\circ N_4^{-1}(\partial)=0
\,,
$$
and multiplying both sides by $N_4(\partial)$ on the right,
we get
\begin{equation}\label{20130607:eq5}
M_2(\partial)
=
M_1(\partial)\circ N_1^{-1}(\partial)\circ 
\big(
N_2(\partial)-D_p^*(\partial)\circ N_4(\partial)
\big)
\,.
\end{equation}


If $M_1(\partial)$ and $N_1(\partial)$
have a common right factor $Q_1(\partial)\in\Mat_{(\ell-m)\times(\ell-m)}\mc V[\partial]$,
then $M(\partial)$ and $N(\partial)$ have the common right factor
$$
Q(\partial)
=
\left(\begin{array}{cc}
Q(\partial)(\partial) & 0 \\
0 & \id_m
\end{array}
\right)
\,\in\Mat_{\ell\times\ell}\mc V[\partial]
\,.
$$
On the other hand, 
if $M(\partial)=\tilde{M}(\partial)\circ Q(\partial)$
and $N(\partial)=\tilde{N}(\partial)\circ Q(\partial)$,
and the association relation \eqref{20130607:eq1} holds,
then it also holds after replacing $M$ by $\tilde{M}$,
$N$ by $\tilde{N}$, and $F$ by $Q(\partial)F\in\mc K^\ell$.
Therefore, we can assume, without loss of generality, that $M_1(\partial)$
and $N_1(\partial)$ are right coprime in the principal ideal ring 
$\Mat_{(\ell-m)\times(\ell-m)}\mc K[\partial]$, \cite{CDSK13a}.


We next use the following simple result:
\begin{lemma}\label{20130611:lem}
Let $A(\partial),B(\partial)\in\Mat_{n\times n}\mc K[\partial]$ 
be right coprime matrix differential operators,
with $B(\partial)$ non-degenerate,
and let $C(\partial)\in\Mat_{n\times r}\mc K[\partial]$.
Then
$A(\partial)\circ B^{-1}(\partial)\circ C(\partial)\in\Mat_{n\times r}\mc K[\partial]$
if and only if
$B^{-1}(\partial)\circ C(\partial)\in\Mat_{n\times r}\mc K[\partial]$.
\end{lemma}
\begin{proof}
The if part is obvious. Conversely, suppose that $A(\partial)\circ B^{-1}(\partial)\circ C(\partial)$
is a matrix differential operator.
Since, by assumption, $A(\partial)$ and $B(\partial)$ are right coprime,
and since $\Mat_{n\times n}\mc K[\partial]$ is a principal ideal ring,
we have the Bezout identity
$$
U(\partial)\circ A(\partial)+V(\partial)\circ B(\partial)=\id_{n}
\,,
$$
for some
$U(\partial),V(\partial)\in\Mat_{n\times n}\mc K[\partial]$.
Multiplying both sides of this equation on the right by $B^{-1}(\partial)\circ C(\partial)$,
we get
$$
\begin{array}{l}
\displaystyle{
\vphantom{\Big(}
B^{-1}(\partial)\circ C(\partial)
} \\
\displaystyle{
\vphantom{\Big(}
=
U(\partial)\circ A(\partial)\circ B^{-1}(\partial)\circ C(\partial)
+
V(\partial)\circ C(\partial)
\,,}
\end{array}
$$
which obviously lies in $\Mat_{n\times n}\mc K[\partial]$.
\end{proof}

Applying Lemma \ref{20130611:lem} to
$A(\partial)=M_1(\partial)$, $B(\partial)=N_1(\partial)\,\in\Mat_{(\ell-m)\times(\ell-m)}\mc V[\partial]$,
and 
$C(\partial)=N_2(\partial)-D_p^*(\partial)\circ N_4(\partial)\in\Mat_{(\ell-m)\times m}\mc V[\partial]$,
and recalling equation \eqref{20130607:eq5},
we deduce that
\begin{equation}\label{20130611:eq1}
X(\partial)
:=
M_1^{-1}(\partial)\circ M_2(\partial)
=
N_1^{-1}(\partial)\circ 
\big(
N_2(\partial)-D_p^*(\partial)\circ N_4(\partial)
\big)
\,\in\Mat_{(\ell-m)\times m}\mc V[\partial]
\,.
\end{equation}


By the first equation in \eqref{20130607:eq1} 
and equation \eqref{20130603:eq4}, we have
\begin{equation}\label{20130607:eq8}
\begin{array}{l}
\displaystyle{
\vphantom{\Big(}
\frac{\widetilde{\delta}h}{\widetilde{\delta} u}
=
\big(\id_{\ell-m}\,,\,-D_p^*(\partial)\big)
\circ
\frac{\delta h}{\delta u}
=
\big(\id_{\ell-m}\,,\,-D_p^*(\partial)\big)
\circ
\left(\begin{array}{cc}
N_1(\partial) & N_2(\partial) \\
0 & N_4(\partial)
\end{array}
\right) F
} \\
\displaystyle{
\vphantom{\Big(}
=
\big(N_1(\partial)\,,\,N_2(\partial)-D_p^*(\partial)N_4(\partial)\big) F
=
N_1(\partial)
F_1
\,,
}
\end{array}
\end{equation}
where, by \eqref{20130611:eq1}
\begin{equation}\label{20130607:eq10}
F_1
=
\big(\id_{\ell-m}\,,\,X(\partial)\big) F
\in\mc K^{\ell-m}
\,.
\end{equation}
Moreover, by the second equation in \eqref{20130607:eq1} we have
\begin{equation}\label{20130607:eq9}
\vphantom{\Big(}
P_1
=
\big(M_1(\partial)\,,\,M_2(\partial)\big)F
=
M_1(\partial)F_1
\,.
\end{equation}
In the last equality we used \eqref{20130611:eq1} and \eqref{20130607:eq10}.
Equations \eqref{20130607:eq7}, \eqref{20130607:eq8}, and \eqref{20130607:eq9},
imply that $\tint h\ass{A}P_1$ in $\widetilde{\mc V}$, proving (a).


By part (a) we have that $\tint h\ass{A}P_1$ in $\widetilde{\mc V}$,
i.e. there exists $F_1\in\mc V^{\ell-m}$ such that 
$\frac{\widetilde{\delta} h}{\widetilde{\delta} u}=N_1(\partial)F_1$
and $P_1=M_1(\partial)F_1$,
where \eqref{20130607:eq7} is a fractional decomposition  for $A(\partial)$.
In fact, we can assume that \eqref{20130607:eq7} is minimal, \cite{DSK13}.
By passing to the quotient $\mc V/\mc I$, we thus get
\begin{equation}\label{20130611:eq2}
\frac{\widetilde{\delta}\overline{h}}{\widetilde{\delta} u}=\overline{N_1}(\partial)\overline{F_1}
\,\,\text{ and }
\overline{P_1}=\overline{M_1}(\partial)\overline{F_1}
\,.
\end{equation}
On the other hand, by our assumption,
if \eqref{20130607:eq7} is a minimal fractional decomposition,
then $\overline{N_1}(\partial)\in\Mat_{(\ell-m)\times(\ell-m)}(\mc V/\mc I)[\partial]$ 
is non-degenerate,
and therefore 
$H^C(\partial)=\overline{A}(\partial)=\overline{M_1}(\partial)\circ \overline{N_1}^{-1}(\partial)$
is a fractional decomposition for $H^C(\partial)$.
Therefore, by equation \eqref{20130611:eq2} we get that $\tint \overline{h}\ass{H^C}\overline{P_1}$,
proving (b).
\end{proof}
\begin{remark}\label{20130622:rem2}
We believe that the assumption that
$\overline{N_1}(\partial)$ non-degenerate in part (b) of Theorem \ref{20130530:prop} 
is automatically satisfied
(see Remark \ref{20130622:rem1}).
\end{remark}

\section{Dirac reduction for (non-local) Poisson structures and Hamiltonian equations}\label{sec:4}

\subsection{Dirac modified Poisson structure}\label{sec:4.1}

In the present section we describe how the Dirac modification \eqref{dirac}
of a PVA $\lambda$-bracket
becomes in the special case of a non-local Poisson structure on an algebra of differential functions.

Let $\mc V$ be an algebra of differential functions in the variables $u_1,\dots,u_\ell$,
which is a domain,
and let $H(\partial)\in\Mat_{\ell\times\ell}\mc V((\partial^{-1}))$
be a non-local Poisson structure on $\mc V$.
Let $\{\cdot\,_\lambda\,\cdot\}_H$ be the corresponding PVA $\lambda$-bracket on $\mc V$
given by the Master Formula \eqref{20110922:eq1}.
Let, as in Section \ref{sec:2}, $\theta_1,\dots,\theta_m$ be some elements of $\mc V$,
and let $\mc I=\langle\theta_1,\dots,\theta_m\rangle_{\mc V}\subset\mc V$ 
be the differential ideal generated by them.
Consider the following rational matrix pseudodifferential operator
\begin{equation}\label{20130529:eq1}
C(\partial)=D_\theta(\partial)\circ H(\partial)\circ D_\theta^*(\partial)\,
\in\Mat_{m\times m}\mc V((\partial^{-1}))\,,
\end{equation}
where $D_\theta(\partial)$ is the $m\times\ell$ matrix differential operator
of Frechet derivatives of the elements $\theta_i$'s:
\begin{equation}\label{20130529:eq2}
{D_\theta(\partial)}_{\alpha,i}
=\sum_{n\in\mb Z_+}\frac{\partial\theta_\alpha}{\partial u_i^{(n)}}\partial^n
\,\,,\,\,\,\,
\alpha=1,\dots,m,\,i=1,\dots,\ell\,,
\end{equation}
and $D_\theta^*(\partial)\in\Mat_{\ell\times m}\mc V[\partial]$ is its formal adjoint.
We assume that the matrix $C(\partial)$ in \eqref{20130529:eq1}
is invertible in $\Mat_{m\times m}\mc V((\partial^{-1}))$.
\begin{definition}\label{20130529:def}
The \emph{Dirac modification} of the Poisson structure $H\in\Mat_{\ell\times\ell}\mc V((\partial^{-1}))$
by the constraints $\theta_1,\dots,\theta_m$
is the following skewadjoint $\ell\times\ell$ matrix pseudodifferential operator:
\begin{equation}\label{20130529:eq3}
\widetilde{H}^D(\partial)
=
H(\partial)-H(\partial)\circ D_\theta^*(\partial)\circ C^{-1}(\partial)\circ D_\theta(\partial)\circ H(\partial)
\,.
\end{equation}
\end{definition}
\begin{proposition}\phantomsection\label{20130529:prop}
\begin{enumerate}[(a)]
\item
The Dirac modified $\lambda$-bracket \eqref{dirac}
is related to the Dirac modification $\widetilde{H}^D$ 
of the Poisson structure defined in \eqref{20130529:eq3}
by the Master Formula \eqref{20110922:eq1}:
\begin{equation}\label{20130529:eq4}
\{f_\lambda g\}^D_H=\{f_\lambda g\}_{\widetilde{H}^D}
\,\,\,\,
\text{ for all } f,g\in\mc V
\,.
\end{equation}
In particular, the Dirac modification $\widetilde{H}^D(\partial)$
is a non-local Poisson structure on $\mc V$.
\item
All the elements $\theta_\alpha$'s are central for 
the Dirac modified Poisson structure $\widetilde{H}^D$,
i.e. $\widetilde{H}^D(\partial)\circ D_\theta^*(\partial)=0$, 
and $D_\theta(\partial)\circ \widetilde{H}^D(\partial)=0$.
\end{enumerate}
\end{proposition}
\begin{proof}
By the Master Formula \eqref{20110922:eq1} and the definition \eqref{20130529:eq2}
of the Frechet derivative $D_\theta$, we have that 
the matrix elements $C_{\beta\alpha}(\lambda)$ in \eqref{C} are
$$
\begin{array}{l}
\displaystyle{
C_{\beta\alpha}(\lambda)
=\{{\theta_\alpha}_\lambda\theta_\beta\}_H
=
\sum_{\substack{i,j\in I \\ m,n\in\mb Z_+}} 
\frac{\partial\theta_\beta}{\partial u_j^{(n)}}
(\lambda+\partial)^n
H_{ji}(\lambda+\partial)
(-\lambda-\partial)^m
\frac{\partial\theta_\alpha}{\partial u_i^{(m)}}
} \\
\displaystyle{
=
\sum_{i,j\in I} 
{D_\theta(\lambda+\partial)}_{\beta j}
H_{ji}(\lambda+\partial)
{D_\theta^*(\lambda)}_{i\alpha}
\,.}
\end{array}
$$
Hence, the matrix pseudodifferential operator $C(\partial)$ 
defined by \eqref{C} coincides with the matrix $C(\partial)$ in \eqref{20130529:eq1}.
Furthermore, by the definition \eqref{dirac} of the Dirac modified $\lambda$-bracket, we have
$$
\begin{array}{l}
\displaystyle{
\{{u_i}_{\lambda}u_j\}_H^D
=\{{u_i}_{\lambda}u_j\}_H
-\sum_{\alpha,\beta=1}^m
{\{{\theta_{\beta}}_{\lambda+\partial}u_j\}_H}_{\to}
(C^{-1})_{\beta\alpha}(\lambda+\partial)
\{{u_i}_{\lambda}\theta_{\alpha}\}_H
} \\
\displaystyle{
=H_{ji}(\lambda)
-\sum_{\alpha,\beta=1}^m
\sum_{h,k\in I}
H_{jk}(\lambda+\partial)
{D_\theta^*(\lambda+\partial)}_{k\beta}
(C^{-1})_{\beta\alpha}(\lambda+\partial)
} \\
\displaystyle{
\vphantom{\Big(}
{D_\theta(\lambda+\partial)}_{\alpha h}
H_{hi}(\lambda)
= \widetilde{H}^D_{ji}(\lambda)
\,.}
\end{array}
$$
This proves equation \eqref{20130529:eq4}.
The last assertion of part (a) is a consequence of the fact that the Dirac modification
of a PVA $\lambda$-bracket is again a PVA $\lambda$-bracket, by Theorem \ref{prop:dirac}(a).
Part (b) is an immediate consequence of the 
definition \eqref{20130529:eq3} of the Dirac modification $\widetilde{H}^D(\partial)$,
and the definition \eqref{20130529:eq1} of the matrix $C(\partial)$.
\end{proof}

\subsection{Dirac reduced Poisson structure}\label{sec:4.3}

An interesting situation is when the constraints $\theta_1,\dots,\theta_m\in\mc V$ 
are of type \eqref{20130530:eq7} ($m\leq\ell$).
In this case,
if we write $H(\partial)$ in block form as in \eqref{20130530:eq2},
the matrix $C(\partial)$ defined in \eqref{20130529:eq1} is
$$
C(\partial)
=D(\partial)+D_p(\partial)\circ B(\partial)-B^*(\partial)\circ D_p^*(\partial)
+D_p(\partial)\circ A(\partial)\circ D_p^*(\partial)\,.
$$
Let us assume that $C(\partial)$ is invertible in $\Mat_{m\times m}\mc V((\partial^{-1}))$,
so that 
we can construct the Dirac modified Poisson structure $\widetilde{H}^D(\partial)$.
By Proposition \ref{20130529:prop}(b),
the $\theta_\alpha$'s are central elements 
for $\widetilde{H}^D(\partial)$,
which, by Proposition \ref{20130530:prop}(a), has the form
$$
\widetilde{H}^D(\partial)=
\left(\begin{array}{c}
\id_{\ell-m} \\
-D_p(\partial)
\end{array}
\right)
\circ A^D(\partial)\circ
\big(\id_{\ell-m}\,,\,\,-D_p^*(\partial)\big)
\,.
$$
In fact, it is not hard to compute explicitly the matrix $A^D(\partial)$:
\begin{equation}\label{20130530:eq4}
A^D(\partial)=
A(\partial)+(B(\partial)+A(\partial)\circ D_p^*(\partial))\circ
C^{-1}(\partial)\circ(B^*(\partial)-D_p(\partial)\circ A(\partial))
\,.
\end{equation}
Then, by Proposition \ref{20130530:propb}(c),
under some additional mild assumptions,
we have a ``Dirac reduced'' Poisson structure 
on the quotient algebra of differential functions $\mc V/\mc I$
(with modified partial derivatives \eqref{20130530:eq9}),
$H^D(\partial):=({\widetilde{H}^D})^C(\partial)=\overline{A^D}(\partial)$.
Namely, we have the following:
\begin{corollary}\label{20130604:cor}
Let $H(\partial)\in\Mat_{\ell\times\ell}\mc V((\partial^{-1}))$ 
be a Poisson structure on the algebra of differential functions $\mc V$.
Let $\theta_1,\dots,\theta_m\in\mc V$
be elements of the form \eqref{20130530:eq7} ($m\leq\ell$),
and let $\mc I=\langle\theta_1,\dots,\theta_m\rangle_\mc V\subset\mc V$
the differential ideal of $\mc V$
generated by $\theta_1,\dots,\theta_m$.
\begin{enumerate}[(a)]
\item
The Dirac modification $\widetilde{H}^D(\partial)\in\Mat_{\ell\times\ell}\mc V((\partial^{-1}))$
defined by \eqref{20130529:eq3} is a Poisson structure on $\mc V$.
\item
The matrix $A^D(\partial)\in\Mat_{(\ell-m)\times(\ell-m)}\mc V((\partial^{-1}))$ 
defined in \eqref{20130530:eq4}
is a Poisson structure on the algebra of differential functions $\widetilde{\mc V}$
($=\mc V$ with modified partial derivatives).
It is induced by $\widetilde{H}^D(\partial)$ in the sense that
$$
\{f_\lambda g\}_{\widetilde{H}^D(\partial)}^{\mc V}
=
\{f_\lambda g\}_{A^D(\partial)}^{\widetilde{\mc V}}
\,.
$$
\item
Assume that $\mc I\cap R_{\ell-m}=0$,
that $\mc V/\mc I$ is a domain,
and that the induced matrix
$\overline{A^D}(\partial)\in\Mat_{(\ell-m)\times(\ell-m)}(\mc V/\mc I)[\partial]$
is rational.
Then, the matrix 
\begin{equation}\label{20130604:eq5}
H^D(\partial)=\overline{A^D}(\partial)\in\Mat_{(\ell-m)\times(\ell-m)}(\mc V/\mc I)((\partial^{-1}))
\,,
\end{equation} 
is a Poisson structure on the quotient algebra of differential functions $\mc V/\mc I$
(cf. Proposition \ref{20130530:prop}(b)), induced by $A^D(\partial)$.
\end{enumerate}
\end{corollary}
\begin{proof}
Part (a) is stated in Proposition \ref{20130529:prop}(a).
Part (b) follows from Proposition \ref{20130530:propb}(b)
and Proposition \ref{20130529:prop}(b).
Part (c) follows from Proposition \ref{20130530:propb}(c).
\end{proof}
\begin{remark}\label{20130622:rem3}
We believe that the assumption that
$\overline{A^D}(\partial)$ is rational in part (c) of Corollary \ref{20130604:cor} 
automatically holds (see Remark \ref{20130622:rem1}).
\end{remark}

\subsection{Dirac reduction of a Hamiltonian equation}\label{sec:4.4}

Let $\mc V$ be an algebra of differential functions, which is a domain,
and let $H(\partial)\in\Mat_{\ell\times\ell}\mc V((\partial^{-1}))$
be a (non-local) Poisson structure on $\mc V$.
Let $\theta_1,\dots,\theta_m\in\mc V$,
and consider the Dirac modified Poisson structure $\widetilde{H}^D$ on $\mc V$
given by equation \eqref{20130529:eq3}.
Let
\begin{equation}\label{20130530:eq1c}
\frac{du}{dt}=P\,\in\mc V^\ell\,,
\end{equation}
be a Hamiltonian equation associated to the Poisson structure $H$
and a Hamiltonian functional $\tint h\in\mc V/\partial\mc V$.

Suppose that the elements $\theta_\alpha$'s
are constant densities for the Hamiltonian equation \eqref{20130529:eq3},
i.e. $D_\theta(\partial)P=0$.
Naively, to say that equation \eqref{20130530:eq1c} is Hamiltonian for the Poisson structure $H$
and the Hamiltonian functional $\tint h$ means that
$$
``H(\partial)\frac{\delta h}{\delta u}=P"
\,.
$$
Note that the LHS is not well defined unless $H(\partial)$ is a local Poisson structure.
The precise meaning of the above identity is $\tint h\ass{H}P$ (see Section \ref{sec:3.4}).
If we try to apply naively the Dirac modified Poisson structure 
$\widetilde{H}^D(\partial)$ in \eqref{20130529:eq3}
to $\frac{\delta h}{\delta u}$ we get
$$
``\widetilde{H}^D(\partial)\frac{\delta h}{\delta u}=P"
\,,
$$
since $``H(\partial)\frac{\delta h}{\delta u}=P"$ and, by assumption, $D_\theta(\partial)P=0$.
This indicates that equation \eqref{20130530:eq1c}
should be Hamiltonian also for the Dirac modified Poisson structure $\widetilde{H}^D$
and with the same Hamiltonian functional $\tint h$,
i.e. $\tint h\ass{\widetilde{H}^D}P$.
We cannot prove this statement in general, 
due to the way the association relation $\tint h\ass{H}P$ is defined,
in terms of a fractional decomposition for $H$,
but we will check that this is indeed the case in the example that we will consider
in Section \ref{sec:5}.

Suppose now that the constraints $\theta_1,\dots,\theta_m$ 
are of the special form \eqref{20130530:eq7},
that $\mc I\cap R_{\ell-m}=0$,
that $\mc V/\mc I$ is a domain,
and that $\overline{A^D}(\partial)\in\Mat_{(\ell-m)\times(\ell-m)}(\mc V/\mc I)((\partial^{-1}))$
is rational.
Therefore, 
by Proposition \ref{20130530:prop}(b) the quotient $\mc V/\mc I$
is still an algebra of differential functions, and it is a domain,
and, by Corollary \ref{20130604:cor}(c),
we have a well defined Dirac reduced Poisson structure 
$H^D(\partial)=\overline{A^D}(\partial)$ on $\mc V/\mc I$.
It is natural to expect to have a Hamiltonian equation corresponding to \eqref{20130530:eq1c}.
To say that the constraints $\theta_\alpha$'s are constant densities for the
Hamiltonian equation \eqref{20130530:eq1c}
means (cf. Definition \ref{20130530:def}) 
that $D_\theta(\partial)P=0$, i.e.,
recalling \eqref{20130604:eq4},
$$
P
=
\left(\begin{array}{c}
P_1 \\ -D_p(\partial)P_1
\end{array}\right)
=
\left(\begin{array}{c}
\id_{\ell-m} \\ -D_p(\partial)
\end{array}\right)P_1
\,,
$$
where $P_1\in\mc V^{\ell-m}$ denotes the first $\ell-m$ entries of $P\in\mc V^\ell$.
If, as expected, it happens that $\tint h\ass{\widetilde{H}^D}P$,
then, by Proposition \ref{20130604:prop}(b) and Corollary \ref{20130604:cor}(c),
we have $\tint \overline{h}\ass{H^D}\overline{P_1}$.
In other words, the evolution equation in $\mc V/\mc I$,
$$
\frac{du}{dt}
=
\overline{P_1}
\,,
$$
is Hamiltonian also for the Dirac reduced Poisson structure $H^D$,
with the Hamiltonian functional $\tint \overline{h}$.

We summarize the above observations in the following:
\begin{ansatz}\phantomsection\label{20130604:ans1}
\begin{enumerate}[(a)]
\item
Suppose that
 $\theta_1,\dots,\theta_m\in\mc V$ are constant densities
for the Hamiltonian equation \eqref{20130530:eq1c},
and that the matrix $C(\partial)\in\Mat_{m\times m}((\partial^{-1}))$ in \eqref{20130529:eq1} 
is non-degenerate.
Then equation \eqref{20130530:eq1c} should be a Hamiltonian equation
also for the Dirac modified Poisson structure $\widetilde{H}^D$ defined by \eqref{20130529:eq3},
with the same Hamiltonian functional $\tint h$.
\item
Suppose that the claim in part (a) holds.
Suppose, moreover, that the elements $\theta_\alpha$'s
have the special form \eqref{20130530:eq7},
that $\mc I\cap R_{\ell-m}=0$,
and that $\mc V/\mc I$ is a domain.
Then, $\overline{A^D}(\partial)\in\Mat_{(\ell-m)\times(\ell-m)}(\mc V/\mc I)((\partial^{-1}))$
should be a rational matrix (see Remark \ref{20130622:rem1}).
Moreover, the evolution equation in $\mc V/\mc I$,
\begin{equation}\label{20130603:eq8}
\frac{du_i}{dt}=\overline{P_i}
\,\,,\,\,\,\,
i=1,\dots,\ell-m\,,
\end{equation}
is Hamiltonian with respect to the Dirac reduced Poisson structure $H^D$
given by \eqref{20130604:eq5},
and the Hamiltonian functional $\tint \overline{h}\in\mc V$.
\end{enumerate}
\end{ansatz}

Unfortunately, we have no general statement
relating integrals of motion for a Hamiltonian equation \eqref{20130530:eq1}
to integrals of motion for the corresponding Dirac reduced equation \eqref{20130603:eq8}.
In the next Section we discuss the special case
when the integrals of motion are obtained by the so-called Lenard-Magri scheme
of integrability for a bi-Hamiltonian equation.

\section{Dirac reduction of a bi-Hamiltonian hierarchy}\label{sec:8}

\subsection{Reduction of a bi-Poisson structure}\label{sec:8.1}

Let $\mc V$ be an algebra of differential functions 
in the differential variables $u_1,\dots,u_\ell$, which is a domain.
Recall that two (non-local) Poisson structures
$H_0(\partial),H_1(\partial)\in\Mat_{\ell\times\ell}\mc V((\partial^{-1}))$
are said to be \emph{compatible} if $H_0(\partial)+H_1(\partial)$
is also a Poisson structure on $\mc V$.
In this case we say that $(H_0,H_1)$ form a \emph{bi-Poisson structure} on $\mc V$.

Let $(H_0,H_1)$ be a bi-Poisson structure on $\mc V$.
Let $\theta_1,\dots,\theta_m\in\mc V$ be central elements for $H_0$,
and let $\mc I=\langle\theta_1,\ldots,\theta_m\rangle_{\mc V}$
be the the differential ideal generated by them.
If the matrix pseudodifferential operator 
whose symbol is given by \eqref{C}, with the $\lambda$-brackets for $H_1$, is invertible,
then we can consider the Dirac modified $\lambda$-bracket $\{\cdot\,_\lambda\,\cdot\}_1^D$,
and, by Theorem \ref{20130516:thm1},
$\{\cdot\,_\lambda\,\cdot\}_0$ and $\{\cdot\,_\lambda\,\cdot\}_1^D$
are compatible PVA $\lambda$-brackets on $\mc V$.
It is natural to ask what are the corresponding Poisson structures.
This is given by the following

\begin{proposition}\label{20130604:prop2}
Let $(H_0,H_1)$ be a bi-Poisson structure on $\mc V$.
Let $\theta_1,\dots,\theta_m\in\mc V$ be central elements for $H_0$,
and let $\mc I=\langle\theta_1,\ldots,\theta_m\rangle_{\mc V}$
be the the differential ideal generated by them.
Assume that the matrix
\begin{equation}\label{20130529:eq1b}
C(\partial)=D_\theta(\partial)\circ H_1(\partial)\circ D_\theta^*(\partial)\,
\in\Mat_{m\times m}\mc V((\partial^{-1}))\,,
\end{equation}
is invertible,
and consider the Dirac modified Poisson structure
$\widetilde{H}_1^D(\partial)$, defined by \eqref{20130529:eq3}.
\begin{enumerate}[(a)]
\item
We have
$\widetilde{(H_0+H_1)}^D=H_0+\widetilde{H_1}^D$.
In particular,
$H_0$ and $\widetilde{H}_1^D$ are compatible
Poisson structures on $\mc V$.
\end{enumerate}
Suppose also that the elements $\theta_1,\dots,\theta_m\in\mc V$
are of the form \eqref{20130530:eq7} (with $m\leq\ell$).
We can write $H_0$ and $H_1$ in block form as 
(cf. equations \eqref{20130604:eq1} and \eqref{20130530:eq2}):
$$
\begin{array}{l}
\displaystyle{
H_0(\partial)=
\left(\begin{array}{c}
\id_{\ell-m} \\
-D_p(\partial)
\end{array}
\right)
\circ A_0(\partial)\circ
\big(\id_{\ell-m}\,,\,\,-D_p^*(\partial)\big)
\,,} \\
\displaystyle{
H_1(\partial)=
\left(\begin{array}{cc}
A_1(\partial) & B_1(\partial) \\
-B_1^*(\partial) & D_1(\partial)
\end{array}
\right)
\,.}
\end{array}
$$
Let also $A_1^D(\partial)\in\Mat_{(\ell-m)\times(\ell-m)}\mc V((\partial^{-1}))$
be given by \eqref{20130530:eq4},
with $A,B,D$ replaced by $A_1,B_1,D_1$ respectively.
\begin{enumerate}[(a)]
\setcounter{enumi}{1}
\item
The matrices $A_0(\partial)$ and $A_1^D(\partial)$
are compatible Poisson structures on $\widetilde{\mc V}$
($=\mc V$ with modified partial derivatives).
\item
Assume that $\mc I\cap R_{\ell-m}=0$,
that $\mc V/\mc I$ is a domain,
and that the induced matrices 
$\overline{A_0}(\partial), \overline{A_1^D}(\partial)
\in\Mat_{(\ell-m)\times(\ell-m)}(\mc V/\mc I)((\partial^{-1}))$
are rational.
Then, the central reduction $H_0^C(\partial)=\overline{A_0}(\partial)$,
and the Dirac reduction $H_1^D(\partial)=\overline{A_1^D}(\partial)$,
are compatible Poisson structures on $\mc V/\mc I$.
\end{enumerate}
\end{proposition}
\begin{proof}
Part (a) is an immediate consequence of 
the definition \eqref{20130529:eq3} of Dirac modification of a Poisson structure,
and by the assumption that 
$D_\theta(\partial)\circ H_0(\partial)\,\big(=H_0(\partial)\circ D_\theta^*(\partial)\big)=0$.

By Proposition \ref{20130530:propb}(b), $A_0(\partial)$ 
is a Poisson structure on $\widetilde{\mc V}$,
by Corollary \ref{20130604:cor}(b), $A_1^D(\partial)$ 
is also a Poisson structure on $\widetilde{\mc V}$,
and by part (a) they are compatible.
This proves part (b).

Finally, part (c) follows from Proposition \ref{20130530:propb}(c)
and Corollary \ref{20130604:cor}(c).
\end{proof}
\begin{remark}\label{20130622:rem4}
We believe that the assumptions that
$\overline{A_0}(\partial)$ and $\overline{A_1^D}(\partial)$ are rational
in part (c) of Proposition \ref{20130604:prop2} 
are automatically satisfied (see Remark \ref{20130622:rem1}).
\end{remark}

\subsection{Reduced bi-Hamiltonian hierarchy}\label{sec:8.2}

A \emph{bi-Hamiltonian hierarchy} 
with respect to a bi-Poisson structure $(H_0,H_1)$
is, by definition, a sequence of evolution equations
\begin{equation}\label{20130604:eq6}
\frac{du}{dt_n}=P_n\,\in\mc V^\ell
\,\,,\,\,\,\,
n\in\mb Z_+
\,,
\end{equation}
satisfying the following \emph{Lenard-Magri recursive conditions}:
\begin{equation}\label{20130604:eq7}
\tint h_{n-1}\ass{H_1}P_n
\,\,,\,\,\,\,
\tint h_n\ass{H_0}P_n
\,\,\,\,
\text{ for all } n\in\mb Z_+
\,,
\end{equation}
for some Hamiltonian functionals $\tint h_{-1},\tint h_0,\dots\in\mc V/\partial\mc V$.
In this case,
all Hamiltonian functionals $\tint h_n,\,n\geq-1$,
are integrals of motion for all equations of the hierarchy \eqref{20130604:eq6},
in involution with respect to both Poisson structures $H_0$ and $H_1$.
Also, all commutators $[P_m,P_n]$ lie in a finite dimensional space,
(see \cite{DSK13}).
Hence, each of the equations \eqref{20130604:eq6} is integrable, 
provided that the $\tint h_n$'s are linearly independent.

Let $\theta_1,\dots,\theta_m$ be central elements for $H_0(\partial)$
of the form \eqref{20130530:eq7} ($m\leq\ell$).
By Proposition \ref{20130530:prop}(c),
if $\mc I\cap R_{\ell-m}=0$,
the quotient space $\mc V/\mc I$ has a natural structure of
an algebra of differential functions in the variables $u_1,\dots,u_{\ell-m}$,
with modified partial derivatives $\frac{\widetilde{\partial}}{\widetilde{\partial} u_i^{(s)}}$,
and Proposition \ref{20130604:prop2}
states that we can construct a Dirac reduced bi-Poisson structure
$(H_0^C(\partial),H_1^D(\partial))$ on $\mc V/\mc I$,
under some additional mild assumptions.

For $n\in\mb Z_+$, let ${(P_n)}_1\in\mc V^{\ell-m}$ 
be given by the first $\ell-m$ entries of $P_n\in\mc V^\ell$.
By Proposition \ref{20130604:prop}(b)
we have that 
$\tint \overline{h_n}\ass{H_0^C}\overline{(P_n)_1}$.
By Lemma \ref{20130530:lem}(c), 
the $\theta_\alpha$'s are constant densities for all the equations 
of the hierarchy \eqref{20130604:eq6}.
Hence, 
by Ansatz \ref{20130604:ans1}
we expect that 
$\tint \overline{h_{n-1}}\ass{H_1^D}\overline{(P_n)_1}$.
Therefore, we expect to get a ``reduced'' bi-Hamiltonian hierarchy on $\mc V/\mc I$.
Thus we have the following
\begin{ansatz}\label{20130604:ans2}
Let $(H_0,H_1)$ be a bi-Poisson structure on the algebra of differential functions $\mc V$,
and let \eqref{20130604:eq6} be a bi-Hamiltonian hierarchy
satisfying the Lenard-Magri recursive conditions \eqref{20130604:eq7}.
Let $\theta_1,\dots,\theta_m$ be central elements for $H_0(\partial)$
of the form \eqref{20130530:eq7} ($m\leq\ell$).
\begin{enumerate}[(a)]
\item
By Ansatz \ref{20130604:ans1}(a),
we expect to have the association relations
$\tint h_{n-1}\ass{\widetilde{H}_1^D}P_n$,
for all $n\in\mb Z_+$.
\item
Suppose that the claim in part (a) holds.
Suppose, moreover, that $\mc I\cap R_{\ell-m}=0$,
and that $\mc V/\mc I$ is a domain.
Then, 
$\overline{A_0}(\partial), \overline{A_1^D}(\partial)
\in\Mat_{(\ell-m)\times(\ell-m)}(\mc V/\mc I)((\partial^{-1}))$
should be a rational matrices (see Remark \ref{20130622:rem4}).
Moreover, we have the \emph{Dirac reduced} Lenard-Magri recursive conditions in $\mc V/\mc I$:
$$
\tint \overline{h_{n-1}}\ass{H_1^D}\overline{(P_n)_1}
\,\,,\,\,\,\,
\tint \overline{h_n}\ass{H_0^C}\overline{(P_n)_1}
\,\,\,\,
\text{ for all } n\in\mb Z_+
\,.
$$
Hence, we have a bi-Hamiltonian hierarchy
$$
\frac{du}{dt_n}=\overline{(P_n)_1}\,\in(\mc V/\mc I)^\ell
\,\,,\,\,\,\,
n\in\mb Z_+
\,,
$$
which is integrable provided that the local functionals $\tint\overline{h_n}$'s are linearly
independent.
\end{enumerate}
\end{ansatz}

\begin{remark}\label{20130715:rem2}
Using the results of \cite{CDSK13c}
one can show that, in fact, the naive argument used to ``prove''
the association relation $\tint h_{n-1}\ass{\widetilde{H}_1^D}P_n$
before Ansatz \ref{20130604:ans1},
can be made into a formal proof under the assumption that \eqref{20130529:eq3}
is a minimal rational expression for the Dirac modified Poisson structure $\widetilde{H}_1^D$,
i.e.
\begin{equation}\label{20130715:eq1}
\sdeg(\widetilde{H}_1^D)=2\sdeg(H_1)+\sdeg(C^{-1})
\,.
\end{equation}
(See \cite{CDSK13c} for the definition of the \emph{singular degree} $\sdeg(H)$ of a rational matrix
pseudodifferential operator $H$.)
Under such assumption,
Anstats \ref{20130604:ans1} and \ref{20130604:ans2}
can be made into Theorems.
\end{remark}

\section{Example: Dirac reduction of the generalized Drinfeld-Sokolov hierarchy 
for the minimal nilpotent element of \texorpdfstring{$\mf{sl}_3$}{sl3}}\label{sec:5}

In \cite{DSKV12} we introduced the classical $\mc W$-algebras
associated to a simple Lie algebra $\mf g$ and its nilpotent element
via the classical Hamiltonian reduction of Poisson vertex algebras,
and we studied the associated generalized Drinfeld-Sokolov hierarchies.
In \cite{DSKV13} we considered in detail the example of $\mf g=\mf{sl}_3$ and its minimal
nilpotent element.
Recall that in this case the $\mc W$-algebra is the algebra $R(L,\psi_+,\psi_-,\varphi)$
of differential polynomials in the variables
$L,\psi_+,\psi_-,\varphi$ over $\mb F$,
endowed with two compatible PVA structures
$\{\cdot\,_\lambda\,\cdot\}_0$ and $\{\cdot\,_\lambda\,\cdot\}_1$.
The $\lambda$-brackets among the generators are as follows.
$\varphi$ is central for $\{\cdot\,_\lambda\,\cdot\}_0$, and
$$
\begin{array}{l}
\vphantom{\Big(}
{\{L_\lambda L\}_0}=-2\lambda
\,\,,\,\,\,\,
\{L_\lambda\psi_\pm\}_0=\{{\psi_\pm}_\lambda\psi_\pm\}_0=0
\,\,,\,\,\,\,
\{{\psi_+}_\lambda\psi_-\}_0=1
\,,
\end{array}
$$
and
$$
\begin{array}{l}
\vphantom{\Big(}
{\{L_\lambda L\}_1}=(\partial+2\lambda )L-\frac12\lambda^3
\,\,,\,\,\,\,
\{L_\lambda\psi_\pm\}_1=(\partial+\frac32\lambda )\psi_\pm
\,\,,\,\,\,\,
\{L_\lambda\varphi\}_1=(\partial+\lambda)\varphi
\,\\
\vphantom{\Big(}
\{{\psi_\pm}_\lambda\psi_\pm\}_1=0
\,\,,\,\,\,\,
\{\varphi_\lambda\varphi\}_1=6\lambda
\,\,,\,\,\,\,
\{{\psi_\pm}_\lambda\varphi\}_1=\pm3\psi_\pm
\,\\
\vphantom{\Big(}
\{{\psi_+}_\lambda\psi_-\}_1=
\frac13\varphi^2-\frac12(\partial+2\lambda)\varphi-L+\lambda^2
\,.
\end{array}
$$
(The others are obtained by skewsymmetry.)
The corresponding Poisson structures $H_0(\partial)$ and $H_1(\partial)$,
given by \eqref{20130613:eq2},
are the skewadjoint 
$4\times 4$ matrix differential operators with coefficients in $\mc W$
with the following block forms
\begin{equation}\label{20130614:eq13}
H_0(\partial)=
\left(\begin{array}{cc}
A_0(\partial) & 0_{3\times1} \\
0_{1\times3} & 0_{1\times1}
\end{array}\right)
\,\,,\,\,\,\,
H_1(\partial)=
\left(\begin{array}{cc}
A_1(\partial) & B_1(\partial) \\
-B_1^*(\partial) & D_1(\partial)
\end{array}\right)
\,,
\end{equation}
where
\begin{equation}\label{20130613:eq3}
\begin{array}{l}
\displaystyle{
\vphantom{\begin{array}{c} 1\\1\\1\\1\end{array}}
A_0(\partial)=
\left(\begin{array}{ccc}
-2\partial & 0 & 0 \\
0 & 0 & -1 \\
0 & 1 & 0 \\
\end{array}\right)
\,\,,\,\,\,\,
B_1(\partial)
=
\left(\begin{array}{c}
\varphi\partial \\
-3\psi_+ \\
3\psi_-
\end{array}\right)
\,\,,\,\,\,\,
D_1(\partial)=6\partial
\,,} \\
\displaystyle{
A_1(\partial)
=
\left(\!\!\!\begin{array}{ccc}
\partial\!\circ\! L\!+\!L\partial\!-\!\frac12\partial^3 &
\frac12\partial\circ\psi_++\psi_+\partial &
\frac12\partial\circ\psi_-+\psi_-\partial \\
\partial\circ\psi_++\frac12\psi_+\partial &
0 & 
\Big(\!\!\!\begin{array}{c}
-\frac12(\partial\!\circ\!\varphi+\varphi\partial)\\
-\frac13\varphi^2+L-\partial^2
\end{array}\!\!\!\Big) \\
\partial\!\circ\!\psi_-+\frac12\psi_-\partial &
\Big(\!\!\!\begin{array}{c}
-\frac12(\partial\!\circ\!\varphi+\varphi\partial)\\
+\frac13\varphi^2-L+\partial^2 
\end{array}\!\!\!\Big) &
0
\end{array}\!\!\!\right)
\,.}
\end{array}
\end{equation}
The element $\varphi\in\mc W$ is central for $H_0(\partial)$.
Let $\mc I=\langle\varphi\rangle_{\mc W}$
be the differential ideal generated by $\varphi$.
Clearly, $\mc I\cap R(L,\psi_+,\psi_-)=0$,
and the quotient algebra $\mc W/\mc I$ 
is naturally identified with $R(L,\psi_+,\psi_-)$,
the algebra of differential polynomials in the variables $L,\psi_+,\psi_-$.
In particular, it is a domain.
The Frechet derivative of $\varphi$ is $D_\varphi(\partial)=(0\,0\,0\,1)$.
Hence, the matrix \eqref{20130529:eq1b} is $C(\partial)=D_1(\partial)=6\partial$,
which is invertible in $\mc W((\partial^{-1}))$.
By Proposition \ref{20130604:prop2}(b),
the matrices 
$$
A_0(\partial)
\,\,\text{ and }\,\,
A_1^D(\partial)=A_1(\partial)+B_1(\partial)\circ D_1^{-1}(\partial)\circ B_1^*(\partial)\,,
$$
form a compatible pair of Poisson structure on $\widetilde{\mc W}$
($=\mc W$ considered as algebra of differential functions in 
the variables $L,\psi_+,\psi_-$, with $\varphi$ treated as a quasiconstant).
Clearly, 
$A_1^D(\partial)$ is a rational matrix pseudodifferential operator.
We can compute explicitly its fractional decomposition:
$A_1^D(\partial)=A_1(\partial)+M(\partial)\circ N^{-1}(\partial)$,
where
\begin{equation}\label{20130613:eq4}
\begin{array}{l}
M(\partial)
=
\left(\begin{array}{ccc}
0 & 0 & \varphi\partial\circ\psi_-^2 \\
0 & 0 & -3\psi_+\psi_-^2 \\
0 & 0 & 3\psi_-^3
\end{array}\right)
\,,\\
N(\partial)
=
\left(\begin{array}{ccc}
\psi_+^2 & 0 & 0 \\
-\frac13(\psi_+\partial+2\psi_+^\prime)\circ\varphi & \psi_- & 0 \\
0 & \psi_+ & 2(\psi_-\partial+2\psi_-^\prime)
\end{array}\right)
\,.
\end{array}
\end{equation}

Clearly, 
the images $H_0^C(\partial)$ and $H_1^D(\partial)$
of $A_0(\partial)$ and $A_1^D(\partial)$ respectively
in the quotient space $\mc W/\mc I=R(L,\psi_+,\psi_-)$ 
are rational matrix pseudodifferential operators
with coefficients in $R(L,\psi_+,\psi_-)$.
Hence, by Proposition \ref{20130604:prop2}(c),
they form a compatible pair of Poisson structures on $R(L,\psi_+,\psi_-)$.
Explicitly, they are
\begin{equation}\label{20130613:eq5}
H_0^C(\partial)=
\left(\begin{array}{ccc}
-2\partial & 0 & 0 \\
0 & 0 & -1 \\
0 & 1 & 0 \\
\end{array}\right)
\,,
\end{equation}
and
\begin{equation}\label{20130613:eq6}
H_1^D(\partial)
=
\left(\begin{array}{ccc}
\partial\circ L+L\partial-\frac12\partial^3 &
\frac12\partial\circ\psi_++\psi_+\partial &
\frac12\partial\circ\psi_-+\psi_-\partial \\
\partial\circ\psi_++\frac12\psi_+\partial &
\frac32\psi_+\partial^{-1}\psi_+ & 
\Big(\!\!\!\begin{array}{c}
L-\partial^2 \\
-\frac32\psi_+\partial^{-1}\circ\psi_-
\end{array}\!\!\!\Big) \\
\partial\circ\psi_-+\frac12\psi_-\partial &
\Big(\!\!\!\begin{array}{c}-L+\partial^2 \\
-\frac32\psi_-\partial^{-1}\circ\psi_+ 
\end{array}\!\!\!\Big) &
\frac32\psi_-\partial^{-1}\psi_-
\end{array}\right)
\,.
\end{equation}
Since the image of $N(\partial)$ in the quotient space $\mc W/\mc I$
is still non-degenerate,
a fractional decomposition for $H_1^D(\partial)$
is obtained projecting the fractional decomposition \eqref{20130613:eq4} for $A_1^D(\partial)$.
We have
$H_1^D(\partial)=\overline{A_1}(\partial)+\overline{M}(\partial)\circ \overline{N}^{-1}(\partial)$,
where
\begin{equation}\label{20130613:eq7}
\begin{array}{l}
\overline{A_1}(\partial)
=
\left(\begin{array}{ccc}
\partial\circ L+L\partial-\frac12\partial^3 &
\frac12\partial\circ\psi_++\psi_+\partial &
\frac12\partial\circ\psi_-+\psi_-\partial \\
\partial\circ\psi_++\frac12\psi_+\partial &
0 & 
L-\partial^2 \\
\partial\circ\psi_-+\frac12\psi_-\partial &
-L+\partial^2 &
0
\end{array}\!\!\!\right)
\,,\\
\overline{M}(\partial)
=
\left(\begin{array}{ccc}
0 & 0 & 0 \\
0 & 0 & -3\psi_+\psi_-^2 \\
0 & 0 & 3\psi_-^3
\end{array}\right)
\,\,,\,\,\,\,
\overline{N}(\partial)
=
\left(\begin{array}{ccc}
\psi_+^2 & 0 & 0 \\
0 & \psi_- & 0 \\
0 & \psi_+ & 2(\psi_-\partial+2\psi_-^\prime)
\end{array}\right)
\,,
\end{array}
\end{equation}
are the images of $A_1(\partial)$, $M(\partial)$ and $N(\partial)$
in the quotient space $\mc W^D=\mc W/\mc I=R(L,\psi_+,\psi_-)$.

In \cite[Ex.6.4]{DSKV13} we constructed an infinite sequence of linearly independent 
local functionals $\tint g_n,\tint g_{\tilde n}\in\mc W/\partial\mc W$, $n\in\mb Z_+$,
such that $H_0(\partial)\frac{\delta g_{0}}{\delta u}=H_0(\partial)\frac{\delta g_{\tilde{0}}}{\delta u}=0$,
satisfying the Lenard-Magri recursive conditions \eqref{20130604:eq7}.
The first few conserved densities are
\begin{equation}\label{20130614:eq1}
\begin{array}{l}
\displaystyle{
\vphantom{\Big(}
g_0=\widetilde L:=L-\frac1{12}\varphi^2
\,,
\qquad
g_{\tilde0}=\varphi
\,,} \\
\displaystyle{
\vphantom{\Big(}
g_1=\frac12(\psi_+\partial\psi_--\psi_-\partial\psi_+-\varphi\psi_+\psi_-)
-\frac14(L-\frac1{12}\varphi^2)^2
\,,
\qquad
g_{\tilde1}=6\psi_+\psi_-
\,.}
\end{array}
\end{equation}
Hence, we have the corresponding integrable hierarchy of bi-Hamiltonian equations
\begin{equation}\label{20130614:eq2}
\frac{du}{dt_n}=P_n
=H_1(\partial)\frac{\delta g_{n-1}}{\delta u}
=H_0(\partial)\frac{\delta g_{n}}{\delta u}
\,\,,\,\,\,\,
n\in\mb Z_+
\,,
\end{equation}
and similar equations with $n$ replaced by $\tilde n$.

The first few equations of the hierarchy are as follows
\begin{equation}\label{20130614:eq3}
\frac{d}{dt_0}
\left(\begin{array}{c} L \\ \psi_+ \\ \psi_- \\ \varphi \end{array}\right)
=
\left(\begin{array}{c} 
\widetilde{L}^\prime \\ 
\psi_+^\prime + \frac{1}{2}\varphi\psi_+ \\ 
\psi_-^\prime - \frac{1}{2}\varphi\psi_- \\
0 \end{array}\right)
\,,
\qquad
\frac{d}{dt_{\tilde0}}
\left(
\begin{array}{c}
L \\ \psi_+ \\ \psi_- \\ \varphi 
\end{array}\right)
=
\left(
\begin{array}{c}
0 \\ -3\psi_+ \\ 3\psi_- \\ 0
\end{array}
\right)
\,,
\end{equation}
and
\begin{equation}\label{20130614:eq4}
\begin{array}{rcl}
\displaystyle{
\vphantom{\Big(}
\frac{dL}{dt_1}
}&=&\displaystyle{
\frac14\widetilde L^{\prime\prime\prime}-\frac3{2}\widetilde{L}\widetilde{L}^\prime+
\frac3{2}\left(\psi_+\psi_-^{\prime\prime}-\psi_-\psi_+^{\prime\prime}\right)
-\frac{3}{2}(\varphi\psi_+\psi_-)^\prime
\,,} \\ 
\displaystyle{
\vphantom{\Big(}
\frac{d\psi_\pm}{dt_1}
}&=&\displaystyle{
\psi_{\pm}^{\prime\prime\prime}
\pm\frac{3}{2}\varphi\psi_{\pm}^{\prime\prime}
\pm\frac{1}{2}\psi_{\pm}\varphi^{\prime\prime}
\pm\frac{3}{2}\varphi^\prime \psi_{\pm}^\prime
-\frac3{2}\widetilde{L} \psi_{\pm}^\prime
-\frac3{4}\psi_{\pm} \widetilde{L}^\prime
} \\ 
&& \displaystyle{
\vphantom{\Big(}
\mp\frac3{4}\phi\psi_{\pm}\widetilde{L}
+\frac3{4}\phi^2\psi_{\pm}^\prime
+\frac{3}{4}\psi_{\pm}\varphi \varphi^\prime
\pm\frac{3}{2}\psi_{\pm}\psi_+\psi_-
\pm\frac{1}{8}\phi^3\psi_{\pm}
\,,} \\
\displaystyle{
\vphantom{\Big(}
\frac{d\varphi}{dt_1}
}&=&\displaystyle{
0
\,,}\\
\displaystyle{
\vphantom{\bigg)}
\frac{dL}{dt_{\tilde1}}
}
&=&
\displaystyle{
\vphantom{\Big)}
6\left(\psi_+\psi_-\right)^\prime
\,,
}\\
\displaystyle{
\vphantom{\Big)}
\frac{d\psi_{\pm}}{dt_{\tilde1}}
}
&=&
\displaystyle{
\vphantom{\Big(}
\mp3\psi_{\pm}''\pm3L\psi_{\pm}\mp\varphi^2\psi_{\pm}
-\frac32\psi_{\pm}\varphi'-3\varphi\psi_{\pm}'
\,,}
\\
\displaystyle{
\vphantom{\Big(}
\frac{d\varphi}{dt_{\tilde1}}
}
&=&
\displaystyle{
\vphantom{\Big(}
0
\,.}
\end{array}
\end{equation}

We want to prove that, in this example,
all the statements in Ansatz \ref{20130604:ans2} hold.
By the Lenard-Magri recursive conditions \eqref{20130614:eq2}
we have that
\begin{equation}\label{20130614:eq5}
\begin{array}{l}
\displaystyle{
\vphantom{\Big(}
A_1(\partial)\frac{\widetilde{\delta}g_{n-1}}{\widetilde{\delta}u}
+B_1(\partial)\frac{\delta g_{n-1}}{\delta\varphi}
=
A_0(\partial)\frac{\widetilde{\delta}g_{n}}{\widetilde{\delta}u}
=
(P_n)_1
\,,} \\
\displaystyle{
\vphantom{\Big(}
-B_1^*(\partial)\frac{\widetilde{\delta}g_{n-1}}{\widetilde{\delta}u}
+D_1(\partial)\frac{\delta g_{n-1}}{\delta\varphi}
=
0
\,,}
\end{array}
\end{equation}
where, recalling the notation \eqref{20130603:eq4}, 
$\frac{\widetilde{\delta}}{\widetilde{\delta}u}
=\left(\begin{array}{ccc}
\frac{\delta}{\delta L}\,, &
\frac{\delta}{\delta\psi_+}\,, &
\frac{\delta}{\delta\psi_-}
\end{array}\right)^T$,
and, as usual, $(P_n)_1\in\mc W^3$ is given by the first three components of $P_n\in\mc W^4$.
We can write explicitly the second equation in \eqref{20130614:eq5} 
using the definitions \eqref{20130613:eq3} of $B_1(\partial)$ and $D_1(\partial)$:
\begin{equation}\label{20130614:eq6}
\partial\Big(\varphi\frac{\delta g_{n-1}}{\delta L}\Big)
+3\psi_+\frac{\delta g_{n-1}}{\delta\psi_+}
-3\psi_-\frac{\delta g_{n-1}}{\delta\psi_-}
+6\partial\frac{\delta g_{n-1}}{\delta\varphi}
=0
\,.
\end{equation}
It follows from equations \eqref{20130613:eq4} and \eqref{20130614:eq6} that
\begin{equation}\label{20130614:eq7}
\frac{\widetilde{\delta}g_{n-1}}{\widetilde{\delta}u}
=
N(\partial)F
\,,
\end{equation}
where
\begin{equation}\label{20130614:eq8}
F=
\left(\begin{array}{l}
\frac1{\psi_+^2}\frac{\delta g_{n-1}}{\delta L} \\
\frac1{3\psi_+\psi_-}\partial\Big(\varphi\frac{\delta g_{n-1}}{\delta L}\Big)
+\frac1{\psi_-}\frac{\delta g_{n-1}}{\delta\psi_+} \\
\frac1{\psi_-^2}\frac{\delta g_{n-1}}{\delta\varphi}
\end{array}\right)
\,\in\mc K^3
\,.
\end{equation}
($\mc K$ denotes the field of fractions of $\mc W$.)
On the other hand,
by the definitions \eqref{20130613:eq3} and \eqref{20130613:eq4} of the matrices $B_1(\partial)$
and $M(\partial)$, we have
\begin{equation}\label{20130614:eq9}
M(\partial)F=
B_1(\partial)\frac{\delta g_{n-1}}{\delta\varphi}
\,.
\end{equation}
By the first equation in \eqref{20130614:eq5} and equation \eqref{20130614:eq9}
we thus get
\begin{equation}\label{20130614:eq10a}
\big(A_1(\partial)\circ N(\partial)+M(\partial)\big)F=
(P_n)_1
\,.
\end{equation}
Recalling the fractional decomposition \eqref{20130613:eq4}
for $A_1^D(\partial)$,
equations \eqref{20130614:eq7} and \eqref{20130614:eq10a}
exactly say that we have the association relations
\begin{equation}\label{20130614:eq10b}
\tint g_{n-1}\ass{A_1^D}(P_n)_1
\,\,\text{ in }\,\,
\widetilde{\mc W}
\,,
\end{equation}
for all $n\in\mb Z_+$.
Namely, the first statement of Ansatz \ref{20130604:ans2} holds.

Next, we go to the quotient algebra $\mc W/\mc I=R(L,\psi_+,\psi_-)$.
By equations \eqref{20130614:eq7}, \eqref{20130614:eq8} and \eqref{20130614:eq10b}
we have
$$
\frac{\widetilde{\delta}\overline{g_{n-1}}}{\widetilde{\delta}u}
=
\overline{N}(\partial)\overline{F}
\,\,,\,\,\,\,
\big(\overline{A_1}(\partial)\circ \overline{N}(\partial)+\overline{M}(\partial)\big)\overline{F}=
\overline{(P_n)}_1
\,,
$$
where
$$
\overline{F}
=
\left(\begin{array}{l}
\frac1{\psi_+^2}\frac{\delta g_{n-1}}{\delta L} \\
\frac1{\psi_-}\frac{\delta g_{n-1}}{\delta\psi_+} \\
\frac1{\psi_-^2}\frac{\delta g_{n-1}}{\delta\varphi}
\end{array}\right)
\,\in\mc (K^D)^3
\,,
$$
where $\mc K^D$ is the field of fractions of $\mc W^D=R(L,\psi_+,\psi_-)$.
Hence,
recalling the fractional decomposition \eqref{20130613:eq7} of $H_1^D(\partial)$,
we conclude that
\begin{equation}\label{20130614:eq12}
\tint \overline{g_{n-1}}\ass{H_1^D}\overline{(P_n)}_1
\,\,\text{ in }\,\,
\mc W^D
\,.
\end{equation}
On the other hand, by the block form \eqref{20130614:eq13} of $H_0(\partial)$
and the Lenard-Magri recursive conditions \eqref{20130614:eq2},
we have that
$$
A_0(\partial)\frac{\widetilde{\delta}g_n}{\widetilde{\delta}u}=(P_n)_1
\,,
$$
and therefore, going to the quotient,
\begin{equation}\label{20130614:eq14}
\tint \overline{g_n}\ass{H_0^C}\overline{(P_n)}_1
\,\,\Big(=H_0^C(\partial)\frac{\widetilde{\delta}\overline{g_n}}{\widetilde{\delta}u}\Big)
\,\,\text{ in }\,\,
\mc W^D
\,.
\end{equation}
The association relations \eqref{20130614:eq12} and \eqref{20130614:eq14}
say that the Lenard-Magri scheme of integrability
still holds after Dirac reduction.
Hence, as stated in Ansatz \ref{20130604:ans2}(b),
we have the integrable bi-Hamiltonian hierarchy
$\frac{du}{dt_n}=\overline{(P_n)}_1$, $n\in\mb Z_+$, in $\mc W^D$,
with constant densities $\overline{g_n}\in\mc W^D$, $n\in\mb Z_+$.
We argue similarly for $n$ replaced by $\tilde n$.

The first few constant densities and equations of the hierarchy are obtained 
taking the images of \eqref{20130614:eq1}, \eqref{20130614:eq3} and \eqref{20130614:eq4}:
$$
\overline{g_0}=L
\,,\,\,\,\,
\overline{g_1}=\frac12(\psi_+\partial\psi_--\psi_-\partial\psi_+)-\frac14L^2
\,,\,\,\,\,
\overline{g_{\tilde1}}=6(\psi_+\psi_-)'
\,,
$$
$$
\frac{d}{dt_0}
\left(\begin{array}{c} L \\ \psi_+ \\ \psi_- \end{array}\right)
=
\left(\begin{array}{c} 
L \\ 
\psi_+ \\ 
\psi_-
\end{array}\right)^\prime
\,,
\qquad
\frac{d}{dt_{\tilde0}}
\left(\begin{array}{c} L \\ \psi_+ \\ \psi_- \end{array}\right)
=
\left(\begin{array}{c} 0 \\ -3\psi_+ \\ 3\psi_- \end{array}\right)
\,,
$$
and
\begin{equation}\label{20130614:eq16}
\frac{d}{dt_1}
\left(\begin{array}{c} L \\ \psi_+ \\ \psi_- \end{array}\right)
=
\left(\begin{array}{c} 
\frac14 L^{\prime\prime\prime}-\frac3{2}LL^\prime+
\frac3{2}\left(\psi_+\psi_-^{\prime\prime}-\psi_-\psi_+^{\prime\prime}\right)
\\ 
\psi_+^{\prime\prime\prime}
-\frac3{2}L \psi_+^\prime
-\frac3{4}\psi_+ L^\prime
+\frac{3}{2}\psi_+^2\psi_-
\\
\psi_-^{\prime\prime\prime}
-\frac3{2}L \psi_-^\prime
-\frac3{4}\psi_- L^\prime
-\frac{3}{2}\psi_+\psi_-^2
\end{array}\right)
\,,
\end{equation}
\begin{equation}\label{20130614:eq16-1}
\frac{d}{dt_{\tilde1}}
\left(\begin{array}{c} L\\ \psi_+ \\ \psi_- \end{array}\right)
=\left(
\begin{array}{c}
6\left(\psi_+\psi_-\right)'
\\
-3\psi_+''+3L\psi_+
\\
3\psi_-''-3L\psi_-
\end{array}
\right)
\,.
\end{equation}

In order to prove integrability of the bi-Hamiltonian hierarchy
$\frac{du}{dt_n}=\overline{(P_n)}_1$,
$\frac{du}{dt_{\tilde n}}=\overline{(P_{\tilde n})}_1$, $n\in\mb Z_+$,
we  prove linear independence of the elements $\overline{(P_n)}_1\in(\mc W^D)^3$.
It is not difficult to show, using the recursive equation \eqref{20130614:eq2},
that 
$$
(\overline{P_n})_1=
\left(\begin{array}{c}-2^{-2n}L^{(2n+1)} \\ \psi_+^{(2n+1)} \\ \psi_-^{(2n+1)} \end{array}\right)+
\text{ terms of lower differential order}\,.
$$
In particular, the elements $\overline{(P_n)}_1$ are linearly independent.
Thus, \eqref{20130614:eq16} and \eqref{20130614:eq16-1}
are the first non-trivial equations of an integrable hierarchy of bi-Hamiltonian equations 
with respect to the Poisson structures \eqref{20130613:eq5} and \eqref{20130613:eq6}
and Hamiltonian functionals $\tint \overline{g_n}$,
$\tint \overline{g_{\tilde n}}$.


%
\end{document}